\documentclass[twoside]{IEEEtran}

\usepackage{amsmath,amssymb,amsthm,amsfonts,dsfont,bbm}
\usepackage[pdftex]{graphicx}
\usepackage{color,setspace}

\theoremstyle{plain}
\newtheorem{thm}{Theorem}
\newtheorem*{theorem*}{Theorem}

\newtheorem{lemma}{Lemma}
\newtheorem{corollary}{Corollary}
\newtheorem*{corollary*}{Corollary}
\theoremstyle{definition}

\newcommand{\bydef}{\stackrel{\bigtriangleup}{=}}

\usepackage{color}
\newcommand{\corr}{\color{black}}

\begin{document}
\title{Epidemic Spreading with External Agents \thanks{An earlier version of this work appeared in the proceedings of IEEE Infocom, Shanghai, China, April 2011 \cite{gopban11}.}}

\author{Siddhartha~Banerjee,~\IEEEmembership{Member,~IEEE},
        Aditya~Gopalan,
        Abhik~Kumar~Das and
        Sanjay~Shakkottai~\IEEEmembership{Fellow,~IEEE}
\thanks{S. Banerjee is with the Department of Management Science and Engineering, Stanford University, Stanford, CA-94305, USA (e-mail: sidb@stanford.edu).}%
\thanks{A. Gopalan is with the Department of Electrical Engineering, Technion- Israel Institute of Technology, Haifa, Israel (e-mail: aditya@ee.technion.ac.il).}%
\thanks{A. K. Das is with Samsung Mobile Solutions Lab, San Diego, CA-92122, USA (e-mail: akdas@utexas.edu).}%
\thanks{S. Shakkottai is with the Department of Electrical and Computer Engineering, The University of Texas at Austin, Austin, TX-78705, USA (e-mail: shakkott@austin.utexas.edu).}%
}


\maketitle

\begin{abstract}
We study epidemic spreading processes in large networks, when the spread is assisted by a small number of \emph{external agents}: infection sources with bounded spreading power, but whose movement is unrestricted vis-{\`a}-vis the underlying network topology. For networks which are \emph{`spatially constrained'}, we show that the spread of infection can be significantly speeded up even by a few such external agents infecting \emph{randomly}. Moreover, for general networks, we derive upper-bounds on the order of the spreading time achieved by certain simple (random/greedy) external-spreading policies. Conversely, for certain common classes of networks such as line graphs, grids and random geometric graphs, we also derive lower bounds on the order of the spreading time over all (potentially network-state aware and adversarial) external-spreading policies; these adversarial lower bounds match (up to logarithmic factors) the spreading time achieved by an external agent with a random spreading policy. This demonstrates that random, state-oblivious infection-spreading by an external agent is in fact order-wise optimal for spreading in such spatially constrained networks.
\end{abstract}

\begin{IEEEkeywords}
Epidemic spreading, infection/information dissemination, long-range spreading, percolation, mobility.
\end{IEEEkeywords}

%
\IEEEpeerreviewmaketitle

\section{Introduction}
Various natural and engineered phenomena involve spreading in networks. Rumors/news propagate among people linked by various means of communication; diseases diffuse as epidemics through populations by various modes; plants disperse pollen/seeds, and thus genetic traits, geographically; riots spread across communities; advertisers aim to disseminate information about products through consumer networks; computer viruses and worms, and also software patches, piggyback across computer networks. Understanding how infection/information/innovation can travel across networks through such processes has been a subject of extensive study in disciplines ranging from epidemiology
\cite{ball04:epidemics,andersonmay92:diseasesbook},
sociology \cite{Rogers2003Diffusion,granovetter1973:weakties} and computer science
\cite{kepwhite91:viruses,kempekleintar03:influence}
to physics \cite{satves02:scalefree}, information theory/networking \cite{massganesh05:epidemics,pittel87:rumor,sanghajek07:gossip,grosstse02:mobility,sardim09:gossipmob}
and applied mathematics
\cite{durliu88:contact,kesten93:speed,DraiefMass}. Though many different models have been considered for such processes, they all involve \emph{epidemic dynamics}: propagation via peer-to-peer interactions between the network nodes. In this paper, we consider one-way dissemination or spreading via such epidemic dynamics -- we refer to this as \emph{epidemic spreading}. Our focus however is on understanding the effect of external agents on such epidemic spreading processes.

In many real-world networks, spreading occurs via the interaction of two processes -- $(i)$ a local epidemic spreading process in the network, and $(ii)$ a global infection process due to agents that are external to the network. For instance, in wireless communication, viruses and worms exploit links due to both short-range technologies like Bluetooth and long-range media such as SMS/MMS and the Internet \cite{wang09:spreading,micnob05:mobepidemic}. To paraphrase Kleinberg \cite{kleinberg07:wlessepi}, outbreaks due to such worms are well-modeled by local spreading on a fixed network of nodes in space (\emph{i.e.} short-range Bluetooth transmissions between users) aided by paths through the network (\emph{i.e.} long-range emails and messages through SMS/MMS/Internet). Other examples of multi-scale spreading include human epidemics \cite{ball04:epidemics} and bio-terror attacks \cite{kaplan03:bioterror}, where infection spreads locally through interpersonal interactions, but is aided by long-range human travel, \emph{e.g.}, via airline routes \cite{Colizza+2006}. In all these cases, some form of agency external to the network is responsible for long-range spread; we want to study the effect of this external agency.

To this end, we consider a model for spreading in networks that decomposes into two distinct components -- a basic \emph{intrinsic spread} component in which infection spreads \emph{locally} via epidemic-dynamics on the underlying graph topology, and an additional \emph{external spread} component in which `external agents' (potentially unconstrained by the underlying graph) can carry infection far from its origin, helping it spread \emph{globally}. More specifically, we develop a rigorous framework with which we quantify the effect that a number of \emph{omniscient} (\emph{i.e.}  network-state aware) and \emph{adversarial} (\emph{i.e.} attempting to maximize the rate of infection) external infection agents can have on the spreading time. We stress that the generic terms `intrinsic spread' and `external spread' serve to model a variety of situations involving heterogeneous modes of spreading -- we discuss this in more detail in Section \ref{sec:model}.

Characterizing the impact of external agency on epidemic spread has a twofold utility:
\begin{enumerate}
\item \emph{(Adversarial perspective)} When malicious epidemics threaten to spread via both intrinsic and external means, it becomes important to understand the \emph{worst-case} spreading behavior brought about by external agents, in order to deploy appropriate countermeasures.
\item \emph{(Optimization perspective)} In cases where dissemination is desirable and the external component can be \emph{controlled} -- \emph{e.g.}, in viral advertising \cite{kempekleintar03:influence}, network protocol design \cite{kempekleindem04:spatgossip}, diffusion of innovations \cite{Rogers2003Diffusion}, etc. -- a study of external-agent assisted spreading can help design faster spreading strategies.
\end{enumerate}

\subsection{Main Contributions}
We consider a graph $G = (V,E)$ in which a spreading process starts at a designated node and commences spreading through two interacting dynamics: an \emph{intrinsic epidemic spread}, and an additional \emph{external infection}. We assume all processes evolve in continuous time, and inter-event times are drawn from independent \emph{exponential} distributions, with appropriate rates\footnote{This is in accordance with assumptions in literature \cite{Durrett07}; however, the results easily extend to a discrete time system, with events in each time slot occurring according to independent Bernoulli random variables.}. The metric of interest is the \emph{spreading time} -- the time taken by the process to spread to all nodes.

We assume the intrinsic spread follows the Susceptible-Infected (SI) dynamics \cite{Brauer12,Durrett07} (alternately referred to in literature as \emph{first-passage percolation} \cite{Kesten03}); network nodes start of as being `susceptible' (S), and transition to being `infected' (I) at a rate proportional to the number of infected neighbors. Once infected, a node remains so forever -- this distinguishes one-way spreading processes considered in this work from related epidemic processes such as the SIS/contact-process \cite{massganesh05:epidemics} or the SIR/Reed-Frost epidemics \cite{kempekleintar03:influence,DraiefMass}, where infected nodes can recover. A formal description is provided in Section \ref{sec:model}.

To model the external infection process, we allow \emph{every} node in the graph to get infected at a potentially different (including zero) rate at each instant; thus at time $t$, the state of the network consists of a set of nodes which are infected (therefore determining the intrinsic spreading process) and a $|V|$-dimensional vector $\bar{L}(t)$ of external infection-rates for each node. The spreading power, or \emph{virulence}, of the external agents is measured by $L(t)\triangleq||\bar{L}(t)||_1$, i.e., the sum rate of external infection. Subject to restrictions on the virulence $L(t)$, we allow $\bar{L}(t)$ to be chosen as a function of the network state and history (\emph{omniscience}) and further, chosen \emph{adversarially}, i.e., designed to minimize the spreading time. In Section \ref{sec:model} we discuss how this model generalizes various models for spreading via external sources. 

Our main message is somewhat surprising -- in the above setting, spite of the `adversarial power' external agents have for spreading infection, we show that for common \emph{spatially-constrained} graphs (i.e., having high diameter/low conductance), a simple random strategy is order-optimal (i.e., up to logarithmic factors). More formally, our contributions in this paper are as follows:
\begin{enumerate}
\item We give \emph{upper bounds} on the spreading time (expectation and concentration results) for general graphs when the external infection pattern is \emph{random}, i.e., when every node is susceptible to the same external-infection rate, irrespective of the infection-state and graph topology. The bounds are based on the graph topology (in particular, diameter/conductance of appropriate subgraphs) and a lower bound $L_{\min}(n)$ on the virulence (which we allow to scale with the network size). We also analyze an alternate \emph{greedy} infection policy based on the same graph partitioning scheme, for which we obtain better bounds for the spreading time.
\item For common classes of structured graphs (ring/line graphs, $d$-dimensional grids) and random graphs (geometric random graphs) which have high diameter/low conductance (spatially-constrained), we use first-passage percolation theory \cite{kesten93:speed} to derive \emph{lower bounds} on the order of spreading times (again, both in expectation and w.h.p) over \emph{all} (possibly omniscient and adversarial) external-infection policies. These lower bounds are in terms of the graph topology and an upper bound $L_{\max}(n)$ on the virulence, and match the upper bounds for random spreading up to logarithmic factors, showing that \emph{random external-infection policies are order-wise optimal} for such spatially-constrained graphs. Furthermore, they match exactly for the greedy policies, indicating that these bounds are tight. 
\item Apart from these results, the general bounds (and related techniques) we derive are of independent interest. They provide a fairly complete picture of the dependence of spreading time on external virulence and graph topology in a wide regime; in particular, it is tight for graphs with \emph{polynomially-bounded diameter} (i.e., diameter $D(n)=\Omega(n^{\alpha})$ for some $\alpha>0$) and {sub-linear external virulence} (i.e.,$||L(t)||_1=o(n)$). To demonstrate this, we discuss how other external-infection models (graphs with additional static or dynamic edges, mobile agents) can be analyzed in our framework, and what our bounds translate to in such cases.
\end{enumerate}

\subsection{Related Work}
There is a large body of prior work  -- both analysis as well as design-oriented studies -- which looks at spreading and epidemic processes on networks.

Several authors have studied the behavior and effects of intrinsic epidemic processes on networks, both numerically using data/simulations \cite{satves02:scalefree,Rogers2003Diffusion,granovetter1973:weakties,kepwhite91:viruses}
and analytically \cite{massganesh05:epidemics,durliu88:contact,dshah:gossipbook,DraiefMass}. More relevant to our work on the effect of external agents, several numerical studies have investigated the spread of infectious diseases via specific mobility patterns, \emph{e.g.} under airline networks \cite{Colizza+2006}, heterogeneous geographic means \cite{ball04:epidemics}, and electronic pathways
\cite{kleinberg07:wlessepi,wang09:spreading}. 

Analyzing epidemic behavior is complemented by works seeking to control such processes for various purposes. For example, various authors have looked at designing algorithms for optimal seeding in networks to maximize the spread of an SIR epidemic \cite{kempekleintar03:influence}, for ensuring long-lasting SIS epidemics by distributing virulence across edges \cite{wagner2005designing}, and for efficient routing over spatial networks with fixed long-range links \cite{kosklein08:pathways}. 

In settings where the epidemic is undesirable, the topic of interest is the worst-case behavior of such processes, in particular, when controlled by an adversary. For example, malware has been studied by modeling it as an adversarially-controlled deterministic epidemic on the complete graph \cite{khouzani12}. In stochastic settings, understanding optimal design \cite{wagner2005designing,lelarge2009} also helps characterize the worst-case spread of epidemics -- moreover, this helps in the design of vaccination strategies to counter them in adversarial settings \cite{borgs2010distribute,lelarge2009}.

One-way epidemic spreading/dissemination in networks is an important primitive in communication engineering. Notable studies consider settings where all network nodes are simultaneously mobile -- for designing gossip algorithms \cite{sardim09:gossipmob,kempekleindem04:spatgossip} or improving the capacity of wireless networks \cite{grosstse02:mobility}
-- and analysis of rumor spreading on fully-connected graphs
\cite{pittel87:rumor,sanghajek07:gossip}. The two most-relevant streams of work for our paper are Kesten and others' investigations into first-passage percolation \cite{Kesten03, kesten93:speed}, and Alon's study of deterministic spreading with external-agents in $d$-dimensional hypercubes \cite{AlonCube} -- our work considers the impact such external agents, but in general networks, and wherein the underlying spread is via the SI dynamics (first-passage percolation).

\section{Model for Epidemics with External Agents}
\label{sec:model}

We consider underlying graph $G=(V,E)$, or alternately, a sequence of graphs $G_n = (V_n,E_n)$ indexed by $n$, with the $n$-th graph having $n$ nodes; for ease of notation, we label the nodes in $V$ from $1$ to $n$. For instance, $G_n$ could be the ring graph with $n$ nodes, or a (2-dimensional) $\sqrt{n} \times\sqrt{n}$ grid. For convenience, we will drop the subscript $n$ for all quantities pertaining to the graph $G_n$ when the context is clear.

We model the epidemic spread on underlying graph $G$ (or $G_n$) using a continuous-time \emph{spreading process} $(S(t))_{t \geq 0}$. At each time $t$, $S(t)= (S_1(t),\ldots,S_n(t)) \in \{0,1\}^V$ denotes the `infection state' of the nodes in $V$: $S_i(t) = 0$ indicates that node $i \in V$ is healthy (or `susceptible') at time $t$, while  $S_i(t) = 0$ denotes that it is `infected'. $\mathcal{S}(t)$ denotes the set of infected nodes at time $t$, \emph{i.e.}  $\mathcal{S}(t) \bydef \{i \in V: S_i(t) = 1\}$, and we use $\mathcal{N}(S(t))$ to denote its size. In order to capture the effect of external agents, the evolution of $S(t)$ is assumed to be driven by the following modes of infection spread:

\begin{itemize}
\item {\em Intrinsic infection:} This follows the standard SI dynamics or first-passage percolation process \cite{Kesten03}. Initially, at $t = 0$, all nodes are healthy, except for a single infected node (that can be arbitrarily chosen). Once a node is infected, \emph{it always remains infected}, and infects each of its neighboring susceptible nodes at independent random times which are exponentially distributed with mean $1$. 
\item {\em External infection:} At time $t$, in addition to being infected by its neighbors in $G$, each node $i$ is susceptible to an \emph{external infection} 
with an exponential infection-rate $L_i(t)$. The external infection-rate vector $\bar{L}(t) \equiv (L_i(t))_{i \in V}$ can vary with time $t$ and can depend on the state of the network $S(t)$.
\end{itemize}
The dependence of the external-infection rate $\bar{L}(t)$ on the network state allows us to model a wide range of external infection processes transcending the structure of the underlying network ($G$). For instance:
\begin{enumerate}
\item $\bar{L}(t) = 0$ represents infection occurring only through edges of the underlying graph (the standard \emph{SI dynamics or first-passage percolation process}).
\item Small-world networks: Both Kleinberg \cite{marngu04:kleinberg} and also Watts-Strogatz \cite{wattstro98:smallworld} show that adding a few fixed \emph{long-range} edges onto structured networks can dramatically reduce routing time and diameter. Our model captures the dynamics of infection spreading with $L$ such additional edges, say, by letting $L_i(t)$ be the number of long-range edges incident on node $i$ that have an infected node at the other end at time $t$.
\item Long-range edges over the underlying graph, instead of being drawn in a \emph{static} manner, can be \emph{dynamically} added and deleted as time progresses. For instance, infected nodes can ``throw out'' fresh sets of long-range edges at certain times -- this corresponds to choosing fresh sets of long-range infection targets depending on network state or other parameters.
\item Moving beyond long-range structures, the external infection vector can also be used to model `virtual mobility'; the external infection could be caused by one or several mobile agents, whose position is unconstrained by the graph, and which thus spread infection to various parts of the network with corresponding rates $\bar{L}(t)$. 
\item At an even more abstract level, the external agent can be viewed as an external information source with \emph{bandwidth} $L(t)$, which it can share across nodes of the graph. This can be used to design optimal processes for viral advertising, spread of software updates, etc.
\end{enumerate}

To complete our system description, we term the quantity $L(t)=||\bar{L}(t)||_1$ as the external \emph{virulence} at time $t$. In this work, we restrict ourselves to scenarios where the virulence $L(t)$ is uniformly (i.e., for all $t$) upper and lower bounded by functions $L_{\max}(n), L_{\min}(n)$ respectively, that can scale with the network size $n$. Finally, we define the \emph{spreading time} of the epidemic as $T \bydef \inf \{t \geq 0: S(t)=\mathds{1}_n\}$. Our concerns are both to (a) analyze the spreading time under certain natural external infection strategies, and (b) show universal lower-bounds on the spreading time for common structured networks, over a wide class of external infection strategies.

\emph{General Notation}: We use $\mathbb{Z}$ and $\mathbb{R}$ for the set of integers and reals respectively. We use the standard asymptotic notation ($O$, $\Theta$, $\Omega$, $\omega$ and $o$) to characterize the growth rate of functions\footnote{Briefly: $f(n)=\Omega(g(n))$ (alternately $g(n)=O(f(n))$) implies \emph{there exists some} $k>0$ such that $\forall n>N$ (for some large enough $N$), we have $f(n)\geq kg(n)$, while $f(n)=\omega(g(n))$ (alternately $f(n)=o(g(n))$) implies that for all $n$ large enough, we have $f(n)\geq kg(n)$ \emph{for all} $k>0$}. For random variables $X$ and $Y$, the notation $X \leq_{st} Y$ and $Y \geq_{st} X$ means that $Y$ stochastically dominates $X$, \emph{i.e.}  $\mathbb{P}[Y\geq z] \geq \mathbb{P}[X \geq z]$ for all $z$. Where necessary, we follow the convention that $1/\infty \bydef 0$.

\section{Main Results and Discussion}
\label{sec:results}

We now state our results, and discuss how they translate to different models of externally-aided epidemic spreading. Our results are of two kinds: upper bounds for spreading time for \emph{general graphs} under \emph{specific policies} (in particular, random and greedy spreading policies), and lower bounds under \emph{any policy} for \emph{specific graphs} (in particular, rings/line graphs, $d$-dimensional grids and the geometric random graph); these are representative of graphs which are spatially-constrained, and where our bounds are tight. Our results are in terms of properties of the graph, and the bounds $L_{\min}$ and $L_{\max}$ on the virulence $L(t)$ -- recall that the latter can scale with $n$. We conclude the section with a discussion of the applicability of our bounds and techniques in various settings.

\subsection{Upper Bounds for Specific Policies}
\label{sec:resultsubnd}

Our first main result is an upper-bound on the spreading time (both in expectation and with high probability) of the homogeneous external-infection policy, for a general graph $G$. {\corr Such a policy is equivalent to one in which the (single) external agent chooses a node uniformly at random and starts infecting it; hence we hereafter refer to it as the \emph{random spreading policy}. The following result states that given a uniform partition of $G$, the time taken by random spreading to finish is of the order of the number of pieces or the maximum piece diameter, whichever dominates.}
\begin{thm}[Upper bound: Random Spreading, Diameter version]
\label{thm:genach}
Suppose $||\bar{L}(t)||_1 \geq L_{\min} \geq 0$ for all $t \geq 0$, and suppose $L_i(t)=L(t)/n$ for all $i\in V$ (random spreading). Given graph $G$ and any partition $\Pi(G)=\bigcup_{i=1}^{g(\Pi)} G_{i}$ by $g(\Pi)$ connected subgraphs $G_{i}$, each with size at least $s(\Pi)$ and diameter at least $d(\Pi)$. Then:
\begin{enumerate}
\item (Mean spreading time) $\mathbb{E}[T] \leq h(\Pi)\cdot(\log n+1)$,\\ where $h(\Pi)\equiv\max\left(\frac{n}{s(\Pi)L_{\min}},d(\Pi)\right)$.
\item (Spreading time concentration) If $g(\Pi)\leq cn^\delta$ for some constants $c,\delta > 0$, then for any $\gamma>0$, we have:
\begin{align*}
\mathbb{P}[T\geq\kappa h(\Pi)\log n] \leq c'n^{-\gamma},
\end{align*}
where $\kappa\geq 1+\frac{\gamma}{\delta}$ and $c'=2c^{-\kappa+1}$.
\end{enumerate}
\end{thm}

As a preview as to how to apply this result, consider a line graph on $n$ nodes -- this can be partitioned into $\sqrt{n}$ segments each of length $\sqrt{n}$. Then, by the above result, the random spreading policy takes $O(\sqrt{n}\log n)$ time to infect all nodes. 

Next we obtain a spreading time bound for a \emph{greedy} spreading policy, which we call the \emph{Greedy Subgraph Infection} (or GSI) policy. The policy is based on any given graph partition like that in the above theorem, and is as follows: given partition $\Pi$ with subgraphs $G_i, i\in\{1,2,\ldots,g(n)\}$, they are infected through sequential greedy (as opposed to \emph{homogeneous}) external infection, i.e., $||L(t)||_1 = L_{\min}$, and $L(t)$ is supported on a single node $j(t)$ within any \emph{maximally healthy} subgraph at time $t$ (i.e., one which minimizes $|G_i\cap\mathcal{S}(t)|$). The spreading time of the GSI policy is $O(h(\Pi))$ in expectation and \emph{w.h.p.}, which we state as follows:

\begin{thm}[Upper bound for GSI Policy]
\label{thm:GSI}
Suppose graph $G$ admits a partition $\Pi(G)=\bigcup_{i=1}^{g(\Pi)} G_{i}$ of connected subgraphs $G_{i}$, each with diameter $\leq d(\Pi)$; further, suppose $d(\Pi)\geq \log n$. Then for the Greedy Subgraph Infection policy, we have:
$$\mathbb{E}[T] \leq\max\left(\frac{g(\Pi)}{L_{\min}},4d(\Pi)\right).$$
\end{thm}

Again, applying this to the line graph with $n$ nodes, we now get a spreading time of $O(\sqrt{n})$, which improves on the previous bound by a factor of $\log n$.

Finally we give an alternate bound for the spreading time with random external-agents in terms of a different structural property intimately related to spreading ability in graphs -- the \emph{conductance} (also called the \emph{isoperimetric constant}). The conductance $\Psi(G)$ of a graph $G = (V,E)$ is defined as \[\Psi(G) \bydef \inf_{S \subset V: 1 \leq S \leq \frac{|V|}{2}} \frac{E(S,V\setminus S)}{|S|},\]
where for $A, B \subseteq V$, $E(A,B)$ denotes the number of edges that have exactly one endpoint each in $A$ and $B$. The conductance of a graph is a widely studied measure of how fast a random walk on the graph converges to stationarity \cite{bremaud:mcbook,dshah:gossipbook}. 
Analogous to Theorem \ref{thm:genach}, the next result formalizes the idea that spreading on a graph is dominated by the larger of (a) the number of pieces it can be broken into, and (b) the reciprocal of the piece conductance. 

\begin{thm}[Upper bound: Random Spreading, Conductance version]
\label{thm:condach}
Suppose $||\bar{L}(t)||_1 \geq L_{\min} \geq 0$ for all $t \geq 0$, and suppose $L_i(t)=L(t)/n$ for all $i\in V$ (random spreading). Further, graph $G$ has a partition $\Pi(G)=\bigcup_{i=1}^{g(\Pi)} G_{i}$ of $g(\Pi)$ connected subgraphs, each with $s_{\min}(\Pi)\leq|G_i|\leq s_{\max}(\Pi)$ and conductance $\geq\Psi(\Pi)$. Then:
\begin{enumerate}
\item (Mean spreading time) $\mathbb{E}[T] \leq k(\Pi)\cdot(\log g(\Pi)+1)$,\\ where $k(\Pi) \equiv \max\left(\frac{n}{s_{\min}(\Pi)L_{\min}},\frac{2\log s_{\max}(\Pi)}{\Psi(\Pi)}\right)$.
\item (Spreading time concentration) For any $\kappa>0$, we have:
$$\mathbb{P}[T\geq \kappa\cdot k(\Pi)\log g(\Pi)] \leq \frac{\pi^2}{9\kappa^2(\log g(\Pi))^2}.$$
\end{enumerate}
\end{thm}

\subsection{Results: Lower Bounds for Specific Topologies}
\label{sec:resultslbnd}

Having estimated the spreading time of random and greedy external-infection policies, a natural question that arises at this point is: How do these policies compare with the best (possibly omniscient and adversarial) policy, \emph{i.e.}, with the lowest possible spreading time among \emph{all other infection strategies}? To this end, we show that for certain commonly studied \emph{spatially-limited} networks (i.e., with diameter $\Omega(n^{\alpha})$ for some $\alpha>0$), such as line/ring networks, $d$-dimensional grids and random geometric graphs, random spreading yields the best order-wise spreading time up to a logarithmic factor (and the GSI policy yields the best order-wise spreading time) to spread infection. In particular, for each of these classes of graphs, we establish lower bounds on the spreading time of \emph{any} spreading strategy, that match the upper bounds established in the previous section.

\noindent{\bf Rings/Linear Graphs:} Let $G_n = (V_n,E_n)$ be the ring graph on $n$ nodes -- $V_n \bydef \{v_1,\ldots,v_n\}$, $E_n \bydef\{(v_i,v_j): j-i \equiv 1 \mbox{ (mod $n$)}\}$. By partitioning $G_n$ into $\sqrt{nL_{\min}}$\footnote{For ease of notation, we assume that fractional powers of $n$ take integer value; if not, the bounds can be modified by appropriately taking ceiling/floor.} segments, 
each of length $\sqrt{n/L_{\min}}$, from Theorem \ref{thm:genach} we get:
\begin{corollary}[Spreading time for random external-infection on ring graphs]
\label{corr:ringach}
For the random external-spreading policy on the ring/line graph $G_n$, we have:

\begin{enumerate}
\item $\mathbb{E}[T] \leq\sqrt{\frac{n}{L_{\min}}}\cdot(\log n+1)$,
\item For any $\gamma > 0$, if $L_{\min}\leq n^{\delta}$ for some $\delta>0$, then: 
$$\mathbb{P}\Big[T \geq \left(1+\frac{2\gamma}{1+\delta}\right)\sqrt{\frac{n}{L_{\min}}}\log n\Big] \leq 2n^{-\gamma}.$$
\end{enumerate}
\end{corollary}

Thus the spreading time on an $n$-ring, with random external-infection, is $O(\sqrt{n/L_{\min}}\log n)$, both in expectation and with high probability. We now present a corresponding lower bound, that shows that conversely, the spreading time on a grid or line graph with \emph{any} (possibly omniscient) external-infection strategy must be $\Omega(\sqrt{n})$, both in expectation and with high probability. We state this for $L_{\max}=1$, but later generalize the result when considering $d$-dimensional grids.

\begin{thm}[Lower bound for ring graphs]
\label{thm:ring_lower}
For the ring graph $G_n$ with $n$ nodes, given $L_{\max}\leq 1\,\forall t\geq 0$, then for any external-spreading policy, we have:
\begin{enumerate}
\item $\mathbb{E}[T] \geq \frac{2}{3}\sqrt{n}.$
\item $\mathbb{P}\left[T <\sqrt{n}/8\right]\leq 4e^{-\sqrt{n}/8}.$
\end{enumerate}
\end{thm}

\noindent{\bf $d$-dimensional Grids:} Building on the previous result, we next show that the random spread strategy achieves the order-wise optimal spreading time even on $d$-dimensional grid networks where $d \geq 2$. Given $d$, the $d$-dimensional grid graph $G_n = (V_n,E_n)$ on $n$ nodes is given by $V_n \bydef \{1,2,\ldots,n^{1/d}\}^d$, and $E_n \bydef \{(x,y) \in V_n \times V_n: ||x-y||_1 = 1\}$.
 
Consider a partition of $G_n$ into $(n/L_{\min})^{1/(d+1)}$ identical and contiguous `sub-grids' $G_{n,i}$, $i = 1,\ldots, n^{1/(d+1)}$ (for details, refer to Section \ref{sec:gridlbnd}). With such a partition, an application of Theorem \ref{thm:genach} shows that:
\begin{corollary}[Spreading time for random external-infection on $d$-dimensional grids]
\label{corr:gridach}
For the random external-spreading policy on an $n$-node $d$-dimensional grid $G_n$, we have:
\begin{enumerate}
\item $\mathbb{E}[T]\leq\left(\frac{n}{L_{\min}}\right)^{1/(d+1)}\cdot(\log n+1)$,
\item For any $\gamma > 0$, if $L_{\min}\leq n^{\delta}$ for some $\delta>0$, then: 
$$\mathbb{P}\Big[T \geq \left(1+\frac{\gamma(1+d)}{1+\delta}\right)\mathbb{E}[T]\Big]\leq 2n^{-\gamma}$$
\end{enumerate}
\end{corollary}
\emph{i.e.}, {\em the spreading time with random external-infection on a $d$-dimensional $n$-node grid is $O\left(\left(\frac{n}{L_{\min}}\right)^{1/(d+1)} \log n\right)$}
in expectation and with high probability. 

In contrast, we show that \emph{any} external-infection policy on a grid takes time $\Omega\left(\left(\frac{n}{L_{\max}}\right)^{1/(d+1)}\right)$ to finish infecting all nodes with high probability, and consequently also in expectation, thereby showing the above bound is order-optimal.
\begin{thm}[Lower bound for $d$-dimensional grids]
\label{thm:gridconv}
Let $G_n$ be a symmetric $n$-node $d$-dimensional grid graph. Suppose that $||\bar{L}(t)||_1 \leq L_{\max}=\omega(n)$ for all $t\geq 0$. Then, there exist $c_1, c_2 > 0$, not depending on $n$, such that:
\[  \mathbb{P}\left[T \leq c_1 \left(\frac{n}{L_{\max}}\right)^{\frac{1}{d+1}}\right] = O\left(e^{-c_2\left(\frac{n}{L_{\max}}\right)^{\frac{1}{2d+2}}}\right).\]
Further, if $L_{\max} = O(n^{1-\epsilon})$ for some $\epsilon \in (0,1]$, then:
\[ \mathbb{E}[T] = \Omega \left(\left(\frac{n}{L_{\max}}\right)^{\frac{1}{d+1}}\right).\]
\end{thm}

\noindent{\bf Geometric Random Graphs:} Finally, we shift focus from structured graphs to a popular family of random graphs, widely used for modeling physical networks. The \emph{Geometric Random Graph (RGG)} is a random graph model wherein $n$ points (\emph{i.e.} nodes) are placed i.i.d. uniformly in $[0,1] \times [0,1]$. Two nodes $x, y$ are connected by an edge iff $||x-y|| \leq r_n$, where $r_n$ is often called the \emph{coverage radius}. The RGG $G_n = G_n(r_n)$ consists of the $n$ nodes and edges as above.

It is known that when the coverage radius $r_n$ is above a critical threshold of $\sqrt{\log n/\pi}$, the RGG is connected with high probability \cite{gupta:connectivity}. In our last set of results, we show that similar to before, random spreading on RGGs in this critical connectivity regime is optimal upto logarithmic factors. First, we show with high probability that random spreading finishes in time $O(\sqrt[3]{n} \log n)$:
\begin{thm}[Spreading-time for random external-infection on the RGG]
\label{thm:rggach}
For the planar random geometric graph $G_n(r_n)$, if $r_n \geq
\sqrt{\frac{5(1+\gamma)\log n}{n}}$, for random external spreading, we have:
\begin{enumerate}
\item  If $\gamma\geq \frac{2}{3}$, then:
$\,\,\,\mathbb{E}[T]\leq 2\sqrt[3]{n/L_{\min}} \log n$
\item For any $\gamma>0$, choosing $\kappa\geq 1+3\gamma/(1+\delta)$ we have:
$$\mathbb{P}[T \geq \kappa\sqrt[3]{n/L_{\min}} \log n]\leq 2n^{-\gamma}$$
\end{enumerate}
\end{thm}

Finally, we follow this up with a converse result that shows that no other policy can better this time (order-wise, up to the logarithmic factors) with significant probability. This directly parallels the earlier results about spreading times on $2$-dimensional grids, where random mobile spread exhibits the same optimal order of growth. 
\begin{thm}[Lower bound for the RGG]
\label{thm:rggconv}
For the planar geometric random graph $G_n$ with $r_n = O(\sqrt{\log n / n})$ with a single random initially-infected node, and any spreading policy with $L_{\max} = O(n^{1-\epsilon})$ for some $\epsilon \in (0,1]$, $\exists$ $\beta > 0$ such that: $$\lim_{n \to \infty}\mathbb{P}\left[T \geq \beta \frac{\sqrt[3]{n/L_{\max}}}{\log^{4/3}n}\right] = 1.$$
\end{thm}

\subsection{Discussion and Extensions}
The framework of epidemic spreading with external agents encompasses many known models for epidemic spreading with long-range interactions (as we discussed previously in Section \ref{sec:model}): this is done by appropriately specifying $\bar{L}(t)\in\mathbb{R}_+^{|V|}$ as a function of time $t$, network topology and network-state $S(t)$. For example, the presence of a single additional `static long-range' link $(i,j)\in V^2$ is equivalent to setting $L_i(t)=\beta\mathds{1}_{S_j(t)=1}, L_j(t)=\beta\mathds{1}_{S_i(t)=1}$ and $L_k(t)=0\,\forall k\notin\{i,j\}$ (where $\beta$ is the rate of spreading on the edge). We now discuss the implications of our results and techniques on such models of external infection sources.

\noindent\textbf{Static Links:} To demonstrate our results in the context of a graph overlaid with additional static edges, consider a $d$-dimensional grid with $L(n)$ additional static links. Then we have the following lower-bound for the spreading time $T$ (obtained by setting $L_{\max}=L(n)$ in Theorem \ref{thm:gridconv}).
\begin{corollary}
\label{corr:gridLB}
Let $G_n$ be a symmetric $n$-node $d$-dimensional grid graph, with $L(n)$ additional static links. If $L(n) = O(n^{1-\epsilon})$ for some $\epsilon \in (0,1]$, then $\mathbb{E}[T] = \Omega \left(\left(\frac{n}{L(n)}\right)^{\frac{1}{d+1}}\right).$
\end{corollary}
Note that by combining this with Theorem \ref{thm:GSI}, we can also get the same lower bound on the diameter $D(n)$ of the resultant graph. To see this, observe that by considering the entire graph as a single partition, Theorem \ref{thm:GSI} gives that the spreading time is $O\left(D(n)\right)$, and thus $D(n)=\Omega \left(\left(n/L(n)\right)^{\frac{1}{d+1}}\right)$ by Corollary \ref{corr:gridLB}. One consequence of this is in the context of `small-world graphs' \cite{marngu04:kleinberg,wattstro98:smallworld} wherein the diameter of a $d$-dimensional grid on $n$ nodes is reduced to $\Theta(\log n$) by adding $\Omega(n\log n)$ \emph{random} long-range edges. The usefulness of the above result is to show that it is not possible to obtain such sub-polynomial diameters by adding $O(n^{1-\epsilon})$ edges. 

We note also that this bound is tight. We can see this from the following simple example: partition the graph into $L(n)$ identical segments, and add an edge between a chosen vertex $i$ and a single vertex in each segment. Now for an epidemic starting at node $i$, it is easy to see that the resultant process is equivalent to the $2$-phase spreading process in the proof of Theorem \ref{thm:GSI} (i.e., parallel seeding of clusters followed by local spreading in clusters). Hence, the spreading time for this process is $O\left(\left(n/L(n)\right)^{\frac{1}{d+1}}\right)$. 

\noindent\textbf{Dynamic Links and Mobile Agents:} A more surprising result is obtained by considering spreading on a grid with additional \emph{dynamic} links, i.e., long-range links which can change their endpoints as time progresses. Unlike a static link which can transmit the infection only once (before both its endpoints are infected), such dynamic links can be re-used over time to help spread the infection. However, we now show that dynamic links do not in fact reduce the order of the spreading time.

\begin{corollary}
Let $G_n$ be a symmetric $n$-node $d$-dimensional grid graph, with $L=O(n^{1-\epsilon}), \epsilon\in (0,1]$ additional dynamic links. Then $\mathbb{E}[T] = \Omega \left(\left(\frac{n}{L}\right)^{\frac{1}{d+1}}\right)$.
\end{corollary}

A related model is of epidemic spreading via \emph{mobile agents}--in such a context, assuming $L(n)$ mobile agents, each with constant infection-rate, Theorem \ref{thm:gridconv} again gives the same converse for spreading time, i.e., $\Omega\left(\left(n/L(n)\right)^{\frac{1}{d+1}}\right)$ for $d$-dimensional grids. Furthermore, the techniques of Theorems \ref{thm:genach} and \ref{thm:GSI} can be used to give upper bounds for various models of mobility: for example, for $L$ mobiles moving randomly on a $d$-dimensional grid (where each mobile is unconstrained by the graph as to its next location), Theorem \ref{thm:genach} shows that the spreading time is $O\left(\left(n/L(n)\right)^{\frac{1}{d+1}}\right)$.

\noindent\textbf{Sub-Polynomial Spreading Time:} In the above examples, we consider settings where the spreading time is \emph{polynomial} in the graph size (i.e., $n^{\alpha}$ for some $\alpha\in(0,1]$). However our techniques do not yield tight bounds in the two extreme regions: non-spatially-constrained graphs, i.e., having sub-polynomial diameter, and high external-infection rate, i.e., $L(t)=\Omega(n)$. There is little work in literature in analyzing such regimes, and the existing work focuses on specific graph and infection models. Two notable results in this respect are tight bounds on deterministic spreading with adversarial external-infection in $d$-dimensional hypercubes \cite{AlonCube} (where the graph diameter is $\Theta(\log n)$), and the $\Theta(\log n)$-diameter characterization of small-world graphs \cite{marngu04:kleinberg} (where the number of edges added is $\Omega(n\log n)$); both however use techniques tailored to their specific problems. 

\noindent\textbf{Computational Complexity of Fast-Spreading Policies:} Another interesting set of questions arising from our model concerns the complexity of designing optimal external-infection policies for \emph{general graphs}. This is essentially a Markov Decision Process on the space of all subsets of $V$, and also seems connected to known NP-complete problems (see below). The design of optimal policies is beyond the scope of this work; however, our results indicate that in many relevant settings, simple policies have a good approximation ratio.

Note also that the GSI policy (Theorem \ref{thm:GSI}) takes as input a partition of the graph which balances the number of sets versus the maximum diameter among the sets. A natural question here is whether such a partition could be found easily -- we now briefly point out that this is NP-complete, but does admit a simple constant-factor approximation.

A related problem is one of choosing $k$ seed-nodes from $V$ so as to minimize the maximum distance of any node from one of these seeds. This is a special case of the \emph{$k$-center} problem (with shortest-path distances), which is known to be NP-hard \cite{Feder88}. However it can be easily approximated; in particular, a natural greedy algorithm is known to be $2$-approximate \cite{Feder88}. Our problem of finding a good partition is similar to the $k$-center problem, except that $k$ is now unknown. However, we can still use the $k$-center algorithms for this problem, as follows: Given the $2$-approximate algorithm for $k$-center, we can execute it for values of $k$ chosen sequentially from $\{1,2,4,\ldots,2^{\log n}\}$ -- we stop when the maximum diameter of the resulting partition is less than the current value of $k$. Since the diameter decreases with $k$, it is easy to show that the resulting partition is a $4$-approximation to our problem.

\section{Proofs: Upper Bounds for Specific Policies}

In this section, we formally prove the upper bounds on spreading time we stated in Section \ref{sec:resultsubnd}. We first prove Theorem \ref{thm:genach}, which gives an upper bound for the spreading time achieved by a random external-spreading policy. Essentially, Theorem \ref{thm:genach} says that given any partition of a large graph, the spreading time of an externally-aided epidemic process is determined by $(a)$ the time taken for the spread to start in each segment of the partition and $(b)$ the worst possible time taken by the intrinsic spread within each segment. The former can be estimated under random external-infection using a \emph{coupon-collector} argument, while the latter involves understanding intrinsic epidemic spreading on a graph (i.e., without external aid), using techniques from $(a)$ stochastic majorization and $(b)$ graph sparsification using \emph{shortest-path} spanning trees.

\begin{proof}[Proof of Theorem \ref{thm:genach}]
Under the random external-infection policy, we have that $L_i(t) \geq L_{\min}/n$ for all $i = 1, \ldots, n$. As before,  $(S(t))_{t \geq 0}$ denotes the infection state process. Given a partition $\Pi(G)=\bigcup_{i=1}^{g(\Pi)}G_i$ which divides the nodes into $g(\Pi)$ sets, recall that we define $s(\Pi)=\min_{i}|G_i|$ and $d(\Pi)=\min_{i}\{\mbox{diameter}(G_i)\}$, i.e., the smallest size and diameter of the subsets of the partition. Henceforth for ease of exposition, we suppress the dependence of variables $g,s,d$ on partition $\Pi$, and also use the shorthand $L_{\min}$ for $L_{\min}$.

Observe that each subgraph $G_i$ is prone to infection (i.e., some node in $G_i$ contracts infection) due to external sources with a rate $\geq \frac{L_{\min} \cdot s}{n}$. Now we consider an alternative infection-spreading process $(\tilde{S}(t))_{t \geq 0}$ which evolves in two phases:
\begin{itemize}
\item \emph{Phase-$1$:} Spreading occurs only due to external agents, and not internal epidemic spreading. The phase starts at $t = 0$ and ends when at least one node in each subgraph $G_i$ is infected. Let $T_1$ be the end time of this phase.
\item \emph{Phase-$2$:} Spreading occurs only due to intrinsic epidemic spreading in $G$, and not external sources. At $t=T_1$, for each $G_i$, only the \emph{first node} infected in phase-$1$, say $N_i$, is assumed to be infected, and all other nodes in $G_i$ are considered to be healthy. Finally, the process $\tilde{S}(\cdot)$ proceeds via the SI dynamics \emph{within} each $G_i$, \emph{i.e.} the infection \emph{does not spread} across edges that connect different subgraphs. Denote by $T_2$ the additional time taken (since $T_1$) for all nodes in all the $G_i$ to get infected.
\end{itemize}

Standard coupling arguments (e.g., see Theorem $8.4$ in \cite{DraiefMass}) establish that $\mathcal{N}(S(t))$
stochastically dominates $\mathcal{N}(\tilde{S}(t))$ for all $t$, i.e., $\tilde{S}$ is a 'slower' process than $S$. Thus, the spreading time for $\tilde{S}(\cdot)$ stochastically dominates that of $S(\cdot)$, \emph{i.e.}
\begin{align}
T  &\leq_{st} T_1 + T_2. \label{eqn:stdom}
\end{align}
It remains to estimate the means of $T_1$ and $T_2$, and their tail probabilities, to finish the proof. The analysis for $T_1$ follows a standard coupon-collecting argument: memorylessness of the exponential distribution implies that $T_1$ is stochastically dominated by the maximum of $g$ \emph{i.i.d.} exponential random variables, each with a rate at least $\frac{L_{\min}\cdot s}{n}$. Using a standard result for the expectation of the maximum of \emph{i.i.d.} exponentials, we get:
\begin{align}
\mathbb{E}[T_1] &\leq \frac{nH_{g}}{sL_{\min}}\leq \frac{n(\log g+1)}{sL_{\min}}
\label{eqn:et1}
\end{align}
where $H_k \bydef \sum_{i=1}^k i^{-1} = O(\log k)$ is the $k$th
harmonic number. Also, by a union bound over the tails of $g$
\emph{i.i.d.} exponential random variables, for any $\kappa > 0$ we can estimate the tail of $T_1$:
\begin{align}
\mathbb{P}\left[T_1 \geq \kappa \frac{n\log g}{sL_{\min}}\right] &\leq ge^{-\left(\frac{sL_{\min}}{n}\cdot\frac{\kappa n\log g}{sL_{\min}}\right)}= g^{-\kappa +1}. \label{eqn:tailt1}
\end{align}

To estimate the statistics of $T_2$, we consider the following `slower' mode of spread for phase-$2$: for each subgraph $G_i$, let $W_i$ be a \emph{shortest-path spanning tree} of $G_i$ rooted at the node $N_i$ which is infected in phase-$1$. By our assumption, such a tree has diameter $\leq d$ and can, in principle, be obtained by performing a Breadth-First Search (BFS) on $G_i$ starting at $N_i$. If we now insist that the phase-$2$ static infection process in $G_i$ spreads \emph{only via the edges of} $W_i$, then again, a standard coupling can be used to show that the time $\hat{T}_2$ when all nodes in $G$ get infected thus stochastically dominates $T_2$.

Before proceeding, we need the following simple lemma, which we state without proof:
\begin{lemma}
\label{lem:maxsum}
For real numbers $a_{ij}$, $1 \leq i \leq m$, $1 \leq j \leq n$, $
\max_{i=1}^m \sum_{j=1}^n a_{ij} \leq \sum_{j=1}^n \max_{i=1}^m
a_{ij}$.
\end{lemma}

Now for each tree $W_i$, suppose its leaves are labeled $N_{i1},\ldots,N_{i l(i)}$. Each leaf $N_{ij}$ has a unique path $p_{ij}$ starting from $N_i$ to itself, of length $\leq d$. Let $\hat{T}_{jk}$ be the time taken for the infection to spread across the $k$th edge on this path $p_{ij}$, \emph{i.e.} the (exponentially distributed) interval between the times when the $(k-1)$th node and the $k$th node on the path are infected. Then, the time $\hat{T}_{2,i}$ taken for all nodes in $W_i$ (hence $G_i$) to get infected can be upper-bounded by using Lemma \ref{lem:maxsum}:
\begin{align*}
\hat{T}_{2,i} &= \max_{j=1}^{l(i)} \sum_{k=1}^{|p_{ij}|} \hat{T}_{jk}
\leq \sum_{k=1}^{d} \left(\max_{j=1}^{l(i)} \hat{T}_{jk}\right),
\end{align*}
and a further application of the lemma bounds the phase-$2$ spreading time $\hat{T}_2 = \max_{i=1}^{g} \hat{T}_{2,i}$ as:
\begin{align*}
\hat{T}_2 &\leq \max_{i=1}^{g} \sum_{k=1}^{d}
\left(\max_{j=1}^{l(i)} \hat{T}_{jk}\right)\leq \sum_{k=1}^{d}
\left(\max_{i=1}^{g} \max_{j=1}^{l(i)} \hat{T}_{jk}\right).
\end{align*}
Note that the above inequalities are pointwise, i.e., they hold for every sample-path. The term in brackets is simply the maximum of the infection spread times across all stage-$k$ edges of all the trees $W_i$ within
$G$. Hence, it is stochastically bounded above by the maximum of $n$ i.i.d Exponential($1$) random variables (say $Z_1,\ldots,Z_n$), using which we can write:
\begin{align}
\label{eqn:et2}
\mathbb{E}[T_2] \leq \mathbb{E}[\hat{T}_2] \leq
\sum_{k=1}^{d} H_n\leq d\cdot(\log n+1).
\end{align}
Again, using the union bound to estimate the tail probability of
$T_2$, we have, for any $\kappa > 0$,
\begin{align}
\mathbb{P}[T_2 \geq \kappa d\log n] &\leq \mathbb{P}[\hat{T}_2 \geq \kappa d\log n] \nonumber \\
&\leq d\mathbb{P}[Z_1 \geq \kappa \log n] \nonumber \\
&\leq d\cdot e^{-\kappa \log n } \leq n^{-\kappa+1}. \label{eqn:tailt2} 
\end{align}

We now have all the required pieces. Combining (\ref{eqn:stdom}), (\ref{eqn:et1}) and (\ref{eqn:et2}) with the fact that $g\leq n$ proves the first part of the theorem. For the second part, recall we define $h=\max\left\{\frac{n}{sL_{\min}},d\right\}$. Now if $g\leq cn^{\delta}$ for some constant $c>0$, then (\ref{eqn:tailt1}) gives:
\begin{align*} 
\mathbb{P}[T_1 \geq \kappa h\log n] \leq c'n^{-\delta(\kappa +1)},
\end{align*}
where $c'=c^{-\kappa+1}$. From (\ref{eqn:stdom}) and (\ref{eqn:tailt2}), we get:
\begin{align*} 
\mathbb{P}[T\geq 2\kappa h\log n] &\leq \mathbb{P}[T_1 + T_2 \geq 2\kappa h\log n] \\
&\leq c'n^{-\delta(\kappa +1)} + n^{-\kappa+1} 
\leq 2c'n^{-\delta(\kappa -1)}
\end{align*}
Choosing $\kappa \geq \frac{\gamma}{\delta}+1$ yields the second part of the theorem.
\end{proof}

The factor of $\log n$ in the bound of the Theorem \ref{thm:genach} is in fact \emph{only due to} the `coupon-collector' effect phase-$1$ time $T_1$; a more refined analysis of the phase-$2$ time $T_2$ shows that if $d(\Pi) = \log n+\omega(1)$, \emph{i.e.} the minimum piece diameter is sufficiently large, then $T_2$ is order-wise $d(\Pi)$ in expectation and \emph{w.h.p.} This is the intuition behind the spreading time bound for the \emph{Greedy Subgraph Infection} policy: given the subgraphs $G_i$, they are infected through sequential greedy (as opposed to \emph{homogeneous}) external infection, i.e., $L(t)$ is concentrated on a single node $j(t)$ within any \emph{maximally healthy} subgraph at time $t$, i.e., one which minimizes $|G_i\cap\mathcal{S}(t)|$.
\begin{proof}[Proof of Theorem \ref{thm:GSI}]
Using the same notation as the earlier proof, suppressing dependence of variables on $\Pi$ and $n$. Again, consider the slower, two-phase spreading process, such that $T\leq_{st} T_1 + T_2$: in this case however, phase-$1$ consists of a sequential `seeding' of each subgraph (it is clear that this is stochastically dominated by the greedy subgraph infection). Thus $T_1$ now corresponds to the \emph{sum} of $g$ \emph{i.i.d} exponential random variables with parameter $L_{\min}$ (i.e., there is no coupon-collector effect), and thus, via standard results, concentrates around its mean $\frac{g}{L_{\min}}$. To complete the proof, we need to tighten our previous bound for $\hat{T}_2$ (and hence, $T_2$), which, using our previous notation, can be written as:
\begin{align*}
\hat{T}_2 = \max_{i=1}^{g} \hat{T}_{2,i} = \max_{i=1}^{g} \max_{j=1}^{l(i)} \sum_{k=1}^{|p_{ij}|} \hat{T}_{jk},
\end{align*}
\emph{i.e.}, $\hat{T}_2$ is the maximum sum of infection times over all leaves in all trees $W_i$. Since the total number of leaves in all the trees $W_i$ is at most $n$, a union bound yields, for any $\alpha >0$:
\begin{align*}
\mathbb{P}[\hat{T}_2 > \alpha d] \leq n \mathbb{P}\left[\sum_{i=1}^{d} Z_i > \alpha d\right],
\end{align*}
where all the $Z_i$ are independent Exponential($1$) random
variables. A Chernoff bound yields:
\begin{align*}
\mathbb{P}\left[\sum_{i=1}^{d} Z_i > \alpha d\right] &\leq e^{-\psi \alpha d} \mathbb{E}\left[e^{\psi \sum_{i=1}^{d} Z_i}\right] 
= e^{-\psi \alpha d} (1-\psi)^{-d}
\end{align*}
where $0 \leq \psi < 1$. With $\psi = 1/2$ and any $\alpha > 0$, we have:
\begin{align*}
\mathbb{P}[\hat{T}_2 > \alpha d] \leq n \cdot 2^{d} e^{-\frac{\alpha d}{2}}.
\end{align*}
Finally, for estimating $\mathbb{E}[\hat{T}_2]$ we have,
\begin{align*}
\mathbb{E}[\hat{T}_2]&=\int_0^\infty \mathbb{P}[\hat{T}_2 > x] dx \\
&\leq (2\log 2 + 2)d+ d\int_{2\log 2 + 2}^\infty \mathbb{P}[\hat{T}_2 > \alpha d] d\alpha.\\ 
&\leq 3d(n)+2^{d}nd\int_{2\log 2 + 2}^\infty e^{-\frac{\alpha d}{2}} d\alpha
= 3d + 2ne^{-d},
\end{align*}
and since we have that $d\geq\log n$, we get the result.
\end{proof}
\begin{proof}[Proof of Theorem \ref{thm:condach}]
As in Theorem \ref{thm:genach}, we study an associated two-phase spreading process $(\tilde{S}(t))_{t \geq 0}$, where the first phase takes time $T_1$ to infect at least one node in each $G_i$, and the infection takes a further time $T_2$ to spread within every (connected) $G_i$. Via coupling, we have $T_{\pi_r}\leq_{st} T_1 + T_2$.

As before, $T_1$ is distributed as the maximum of $g$ Exponential random variables, each with rate at least $\frac{s_{\min}L_{\min}}{n}$; thus, for $\kappa>0$, using standard bounds, we have:
\begin{align}
\mathbb{E}[T_1]&\leq\frac{nH_g}{s_{\min}L_{\min}}\leq \frac{n(\log g+1)}{s_{\min}L_{\min}}, \label{eqn:ET1}
\end{align}
and also, for the variance, we have:
\begin{align}
\mbox{Var}[T_1] &\leq \frac{n^2}{s_{\min}^2L_{\min}^2} \sum_{i=1}^{g} \frac{1}{i^2}\leq \frac{\pi^2n^2}{6s_{\min}^2L_{\min}^2}. \label{eqn:VT1}
\end{align}
Next we have that $T_2$ is the maximum of the times $T_{2,i}$ for infection to spread in each subgraph $G_i$. We stochastically dominate each $T_{2,i}$ as follows: for each subgraph $G_i$, consider a continuous time Markov chain $(\hat{Z}_i)_{t \geq 0}$ on the state space $1,\ldots, |V(G_i)|$ with $\hat{Z}_i(0) = 1$ and transitions $j \rightarrow j + 1$ at rate $j\Psi(\Pi)$ (henceforth denoted $\Psi$) if $1 \leq j \leq |V(G_i)|/2$, and at rate $(|V(G_i)| - j)\Psi$ if $|V(G_i)|/2 < j \leq |V(G_i)|-1$. Let $\hat{T}_{2,i}$ be the time taken for the Markov chain $\hat{Z}_i$ to hit its final state $|V(G_i)|$; $\hat{T}_{2,i} = \sum_{j=1}^{|V(G_i)|-1} \hat{T}_{2,i,j}$ where $\hat{T}_{2,i,j}$ is the sojourn time of $\hat{Z}_i$ in state $j$. We claim that $\hat{T}_{2,i}$ stochastically dominates $T_{2,i}$. To see this, note that at any time $t$, if the number of infected nodes in the phase-$2$ spreading process in $G_i$ is $ 1 \leq j \leq |V(G_i)|/2$, then by the definition of conductance, the rate at which a new healthy node in $G_i$ is infected is at least $j\Psi$. Similarly, if the number of infected nodes is $|V(G_i)|/2 < j < |V(G_i)|$ (\emph{i.e.}  the number of healthy nodes is $(|V(G_i)| - j)$), then the rate at which a new healthy node is infected is at least $(|V(G_i)|/2 -j)\Psi$. 
By standard Markov chain coupling arguments (Theorem $8.4$ of \cite{DraiefMass}), we have that $T_{2,i}\leq_{st}\hat{T}_{2,i}$.

By the independence of the original phase-$2$ spreading processes within the $G_i$ for all $i = 1,\dots,g$, we have:
\begin{align*}
T_2 = \max_i T_{2,i} &\leq_{st} \max_i \hat{T}_{2,i} =  \max_i \sum_{j=1}^{|V(G_i)|-1} \hat{T}_{2,i,j}\\
&\leq \sum_{j=1}^{|V(G_i)|-1} \max_i \hat{T}_{2,i,j} 
\end{align*}
Hence we have:
\begin{align}
\mathbb{E}[T_2] &\leq \sum_{j=1}^{|V(G_i)|-1} \mathbb{E}\left[\max_i \hat{T}_{2,i,j}\right] = 2\sum_{j=1}^{|V(G_i)|/2} \frac{\log g}{j\Psi} \nonumber \\
&\leq \frac{2\log s_{\max} \log g}{\Psi}
\label{eqn:boundET2}
\end{align}
And similarly, for the variance, we have: 
\begin{align}
\mbox{Var}\left(\sum_{j=1}^{|V(G_i)|-1} \max_i
\hat{T}_{2,i,j}\right) &= \sum_{j=1}^{|V(G_i)|-1} \mbox{Var}\left(\max_i
\hat{T}_{2,i,j}\right) \nonumber \\
&= 2\sum_{j=1}^{|V(G_i)|/2}\frac{\pi^2}{6j^2 \Psi^2} = \frac{\pi^4}{18\Psi^2}. \label{eqn:boundVT2}
\end{align}
Combining (\ref{eqn:ET1}) and (\ref{eqn:boundET2}) gives the first part of the theorem. Further, recalling $k(\Pi) \equiv \max\left(\frac{n}{sL_{\min}},\frac{\log s}{\Psi}\right)$, we have:
\begin{align*}
&\mathbb{P}[T\geq\kappa k\log g]\leq\mathbb{P}[T_1 + T_2 \geq \kappa k\log g]\\ &\leq\mathbb{P}\left[T_1 + T_2 \geq \frac{\kappa}{2}\left(\frac{n}{s_{\min}L_{\min}}+\frac{\log s_{\max}}{\Psi}\right)\log g\right]
\end{align*}
Now, using the variance estimates (\ref{eqn:VT1}) and (\ref{eqn:boundVT2}) with Chebyshev's inequality, we have for any $\kappa > 0$:
\begin{align*}
&\mathbb{P}[T\geq \kappa k\log g]
\leq\frac{4\mbox{Var}\left(T_1 + \sum_{j=1}^{|V(G_i)|-1} \max_i \hat{T}_{2,i,j}\right)}{\kappa^2\log^2 g \left(\frac{n}{s_{\min}L_{\min}}+\frac{\log s_{\max}}{\Psi}\right)^2} \\
&\leq \frac{\pi^2\left(\frac{n^2}{s_{\min}^2L_{\min}^2}+ \frac{1}{\Psi^2}\right)}{9\kappa^2\log^2 g\left(\frac{n}{s_{\min}L_{\min}}+\frac{\log s_{\max}}{\Psi}\right)^2}\leq \frac{\pi^2}{9\kappa^2(\log g)^2},
\end{align*}
since $\log s_{\max}\geq 1$. This completes the proof.
\end{proof}

\section{Proofs: Lower Bounds for Specific Graphs}

In the previous section, we upper bound the time taken by random/greedy external-infection policies to infect all nodes in a network. Next, we derive corresponding lower bounds for certain commonly studied \emph{spatially limited} networks, in particular, line/ring networks,
$d$-dimensional grids and random geometric graphs. As discussed in Section \ref{sec:resultslbnd}, for each of these classes of graphs, we establish lower bounds on the spreading time of \emph{any} spreading strategy (possibly omniscient and adversarial) that match the upper bounds (upto logarithmic factors for random spread, and exactly for the GSI policy).

\subsection{Ring/Linear Graphs}

As before, for each $n$ we define $G_n = (V_n,E_n)$ to be the ring graph with $n$ contiguous nodes $V_n \bydef \{v_1,\ldots,v_n\}$, $E_n \bydef \{(v_i,v_j): j-i \equiv 1 \mbox{ (mod $n$)}\}$. Partitioning $G_n$ into $\sqrt{nL_{\min}}$ successive segments of length $\sqrt{n/L_{\min}}$, we get (from Theorem \ref{thm:genach}) that {\em the spreading time on an $n$-ring using random external-infection, is $O(\sqrt{n/L_{\min}}\log n)$ in expectation and with high probability} (see Corollary \ref{corr:ringach}).

We now prove that the spreading time on a grid or line graph with \emph{any} (possibly infection-state aware) external-infection spread strategy must be $\Omega(\sqrt{n})$, both in expectation and with high probability. This establishes that for ring graphs (or $1$-dimensional grids), random external-infection is as good as any other form of controlled infection in an order-wise sense, up to logarithmic factors. Furthermore, we use this theorem to introduce a general technique for obtaining lower bounds based on stochastic dominance via a parallel cluster-growing process. For ease of notation, we assume $||\bar{L}(t)||_1 \leq 1$ in this proof -- in the next section, we obtain the more general bound (with dependence on $L_{\max}(n)$) for $d$-dimensional grids.
\begin{proof}[Proof of Theorem \ref{thm:ring_lower}]
To keep the proof general, we use a parameter $\beta$ for the intrinsic spreading rate over an edge (assumed to be $1$ earlier). Along with the spreading process $(S^{\mathcal{P}}(t))_{t \geq 0}$ induced by the policy $\mathcal{P}$, consider a random process $(\tilde{S}(t))_{t \geq 0}$ described as follows:
\begin{enumerate}
\item At all times $t$, $\tilde{S}(t)$ consists of an integer number ($\tilde{C}_t$) of sets of points called \emph{clusters}, where $(\tilde{C}_t)_{t \geq 0}$ is a Poisson process with intensity $L_{\max}=1$, and $\tilde{C}_0 = 1$ (i.e., an `initial' cluster in which intrinsic spreading starts).
\item Once a new cluster is formed at some time $s$, it adds points following a Poisson process of intensity $2\beta$. 
\end{enumerate}

\begin{figure}[b]
  \centering \includegraphics[scale=0.25]{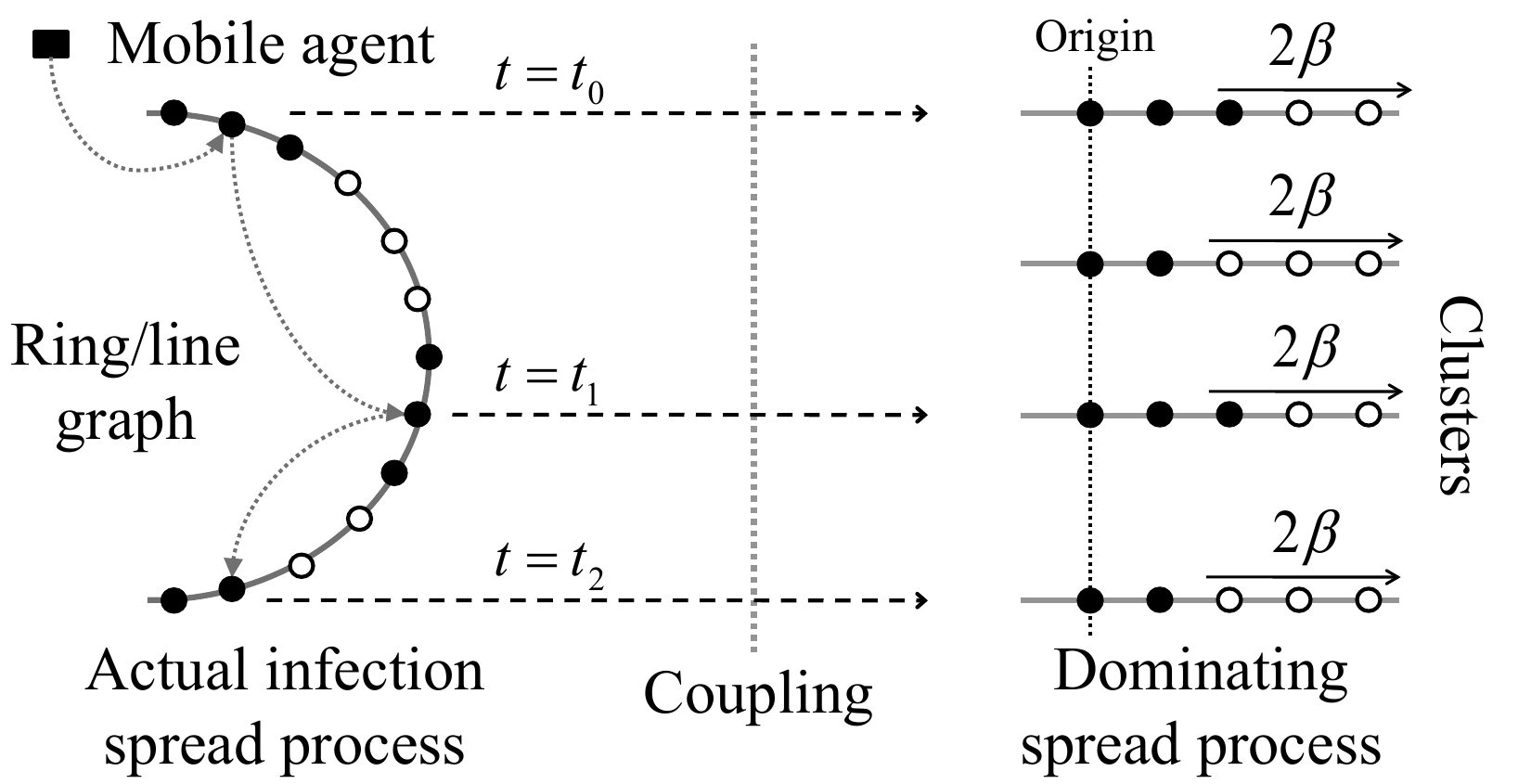}
  \caption{Dominating the infection spread using independently
    growing clusters}
  \label{fig:clusterline}
\end{figure}

Via a coupling argument, it can be shown that for \emph{any} spreading policy $\mathcal{P}$, at all times $t \geq 0$, the total number of points in $\tilde{S}(t)$ (denoted by $\tilde{N}_t$) stochastically dominates that in $S^{\mathcal{P}}(t)$. Informally, this is due to two reasons: first, that the rate of `seeding' of new clusters by $\mathcal{P}$ is at most as fast as that in $\tilde{S}(\cdot)$; secondly, each cluster in $\tilde{S}(\cdot)$ grows independently and without interference from other existing clusters, as opposed to clusters that could `merge' in the process $S^{\mathcal{P}}(\cdot)$. Fig. \ref{fig:clusterline} depicts the structure of the dominating process $\tilde{S}(\cdot)$.

Let $\tilde{T} \bydef \inf \{t \geq 0: \tilde{N}_t = n \}$ be the time when the number of points in $\tilde{S}(\cdot)$ first hits $n$. Owing to the stochastic dominance $\mathcal{N}(S^\pi(t)) \leq_{st}\tilde{N}_t$, we have that for any policy $\mathcal{P}$:
\begin{align}
\tilde{T} &\leq_{st} T_{\mathcal{P}}. \label{eqn:stdomline}
\end{align}  

\noindent Knowing the way $\tilde{S}(\cdot)$ evolves, we can calculate
$E[\tilde{N}_t]$:
\begin{align*}
\mathbb{E}[\tilde{N}_t]=\mathbb{E}[\mathbb{E}[\tilde{N}_t|\tilde{C}_t]]&=\sum_{k=0}^{\infty}\mathbb{P}(\tilde{C}_t=k)\mathbb{E}[\tilde{N}_t|\tilde{C}_t=k]\\
&=\sum_{k=0}^{\infty}\frac{e^{-t}t^k}{k!}\mathbb{E}[\tilde{N}_t|\tilde{C}_t=k].
\end{align*}
Since $\tilde{C}_t$ is a Poisson process, conditioned on
$\{\tilde{C}_t=k\}$, the $k$ cluster-creation instants are distributed uniformly on $[0,t]$. Let the times of these arrivals be $\tilde{T}_1,\ldots,\tilde{T}_{k}$; then $[\tilde{T}_i,t]$ is the time for which the $i$th cluster has been growing. Since every cluster grows at a rate of $2\beta$, conditioned on $\{\tilde{C}_t=k\}$, the expected size of the $i$th cluster is $2\beta(t-\tilde{T}_i)$, $1\leq i\leq k$. Also, the expected size of the `$0$-th' cluster at time $t$ is $2\beta t$. Using $\mathbb{E}[\tilde{T}_i|\tilde{C}_t=k]=t/2$, we obtain:
\begin{align*}
\mathbb{E}[\tilde{N}_t|\tilde{C}_t=k] &= 2\beta t+\sum_{i=1}^{k}\mathbb{E}[2\beta (t-\tilde{T}_i)\arrowvert
\tilde{C}_t=k]\\
&=\beta(k+2)t \\
\Rightarrow \mathbb{E}[\tilde{N}_t] &= \sum_{k=0}^{\infty}\frac{e^{-t}t^k}{k!}\mathbb{E}[\tilde{N}_t|\tilde{C}_t=k]\\
&=\sum_{k=0}^{\infty}\dfrac{e^{-t}t^k}{k!}\beta(k+2)t=\beta t^2+2\beta t.
\end{align*}
Hence, using Markov's inequality, we have:
\begin{align*}
P(\tilde{T}>t) &= P(\tilde{N}_t<n)=1-P(\tilde{N}_t\geq n) \\
&\geq 1-\dfrac{\mathbb{E}[\tilde{N}_t]}{n}\geq 1-\dfrac{\beta (t+1)^2}{n}\\
\Rightarrow\,\,\mathbb{E}[\tilde{T}] &= \int_{0}^{\infty}\mathbb{P}(\tilde{T}>x)dx\geq
\int_{0}^{\sqrt\frac{n}{\beta}-1}\mathbb{P}(\tilde{T}>x)dx \\
&\geq \int_{0}^{\sqrt\frac{n}{\beta}-1}\left(1-\dfrac{\beta
(x+1)^2}{n}\right)dx\\
&=\frac{2}{3}\sqrt \frac{n}{\beta}
-1+\frac{\beta^2}{3n^2}.
\end{align*}
From (\ref{eqn:stdomline}), we have for any policy $\mathcal{P}$, and large enough $n$:
\begin{align*}
\mathbb{E}[T_{\mathcal{P}}]\geq\frac{2}{3\sqrt{\beta}}\sqrt{n} .
\end{align*}

For the second part, denoting the size of the $i^{th}$-created cluster at time $s\geq T_i$ by $\tilde{X}_i(s)$, we can write:
\begin{align*}
&\left(\bigcap_{i=0}^{2et} \{\tilde{X}_i(t+T_i) < 4e\beta
t\}\right) \bigcap \{\tilde{C}_t < 2et \}\\
&\subseteq\left(\bigcap_{i=0}^{\tilde{C}(t)}\{\tilde{X}_i(t+T_i)<4e\beta
t\}\right) \bigcap \{\tilde{C}_t < 2et \} \\
&\subseteq \left(\bigcap_{i=0}^{\tilde{C}(t)} \{\tilde{X}_i(t) < 4e\beta
t\}\right) \bigcap \{\tilde{C}_t < 2et \} \\
&\subseteq \{\tilde{N}_t < 8\beta e^2 t^2 \}.
\end{align*}
Here, the sets refer to sample-trajectories (i.e., points in the underlying sample space) satisfying the stated conditions. Applying a standard Chernoff bound ($\mathbb{P}[Y \geq 2e\lambda] \leq
(2e)^{-\lambda}$ for $Y \sim$ Poisson($\lambda$)) to $\tilde{C}_t
\sim$ Poisson($t$) and $\tilde{X}_i(t+T_i) \sim$ Poisson($2\beta t$) above, we can write:
\begin{align*}
\mathbb{P}[\tilde{N}_t \geq 8\beta e^2 t^2] &\leq
\mathbb{P}[\tilde{C}_t \geq 2et] +
\sum_{i=1}^{2et}\mathbb{P}[\tilde{X}_i(t+T_i) \geq 4e\beta t] \\
&\leq (2e)^{-t} + 2et \cdot (2e)^{-2\beta t} 
\end{align*}
Using the stochastic dominance (\ref{eqn:stdomline}), if $\beta\geq 1$:
\begin{align*}
\mathbb{P}\left[T < \sqrt{\frac{n}{8\beta e^2}}\right] &\leq \mathbb{P}\left[\tilde{T} < \sqrt{\frac{n}{8\beta e^2}}\right]\\
&= \mathbb{P}\left[\tilde{N}_{\sqrt{\frac{n}{8\beta e^2}}} > n\right] 
\leq 4e^{- \sqrt{\frac{n}{8\beta e^2}}}.
\end{align*}
Finally, note that $e^2\leq 8$. This completes the proof.
\end{proof}

\subsection{$d$-Dimensional Grid Graphs}
\label{sec:gridlbnd}

Next, we show that random external-infection spreading achieves the order-wise optimal spreading time (up to logarithmic factors) on $d$-dimensional grid networks for $d \geq 2$, denoted by $G_n=(V_n,E_n)$, with $V_n \bydef \{1,2,\ldots,n^{1/d}\}^d$ and $E_n \bydef \{(x,y) \in V_n \times V_n: ||x-y||_1 = 1\}$.

Consider a partition of $G_n$ into $(nL_{\min}^d)^{1/(d+1)}$ identical and contiguous `sub-grids' $G_{n,i}$, $i = 1,\ldots, (nL_{\min}^d)^{1/(d+1)}$. By this, we mean that each $G_{n,i}$ is induced by a copy of $\{1,2,\ldots,(n/L_{\min})^{1/(d+1)}\}^d$ (and thus has $(n/L_{\min})^{d/(d+1)}$ nodes). For instance, in the case of a planar $\sqrt{n} \times \sqrt{n}$ grid (with $L_{\min}=1$), imagine tiling it horizontally and vertically with $\sqrt[3]{n}$ identical $\sqrt[3]{n} \times \sqrt[3]{n}$ sub-grids (Fig. \ref{fig:2dgrid}). With such a partition, an application of Theorem \ref{thm:genach} (see Corollary \ref{corr:gridach}) shows that {\em the spreading time with random external-infection on a $d$-dimensional $n$-node grid is $O\left(\left(\frac{n}{L_{\min}}\right)^{1/(d+1)} \log n\right)$} in expectation and with high probability.

\begin{figure}[!b]
\centering \includegraphics[scale=0.24]{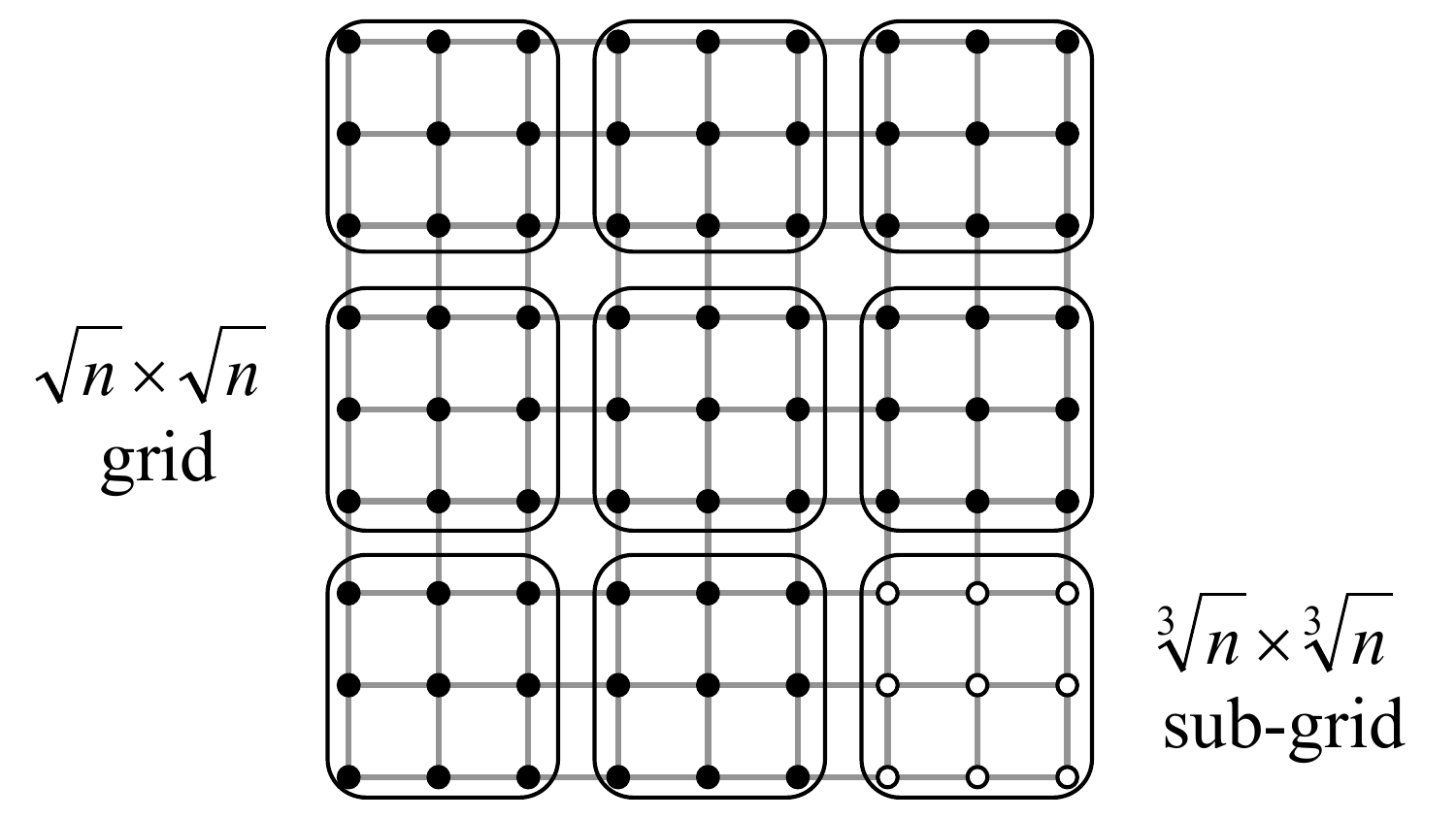}
  \caption{Partitioning a planar grid into sub-grids}
  \label{fig:2dgrid}
\end{figure}

We now show that \emph{any} external-infection spreading policy on a grid must take time $\Omega\left(\left(\frac{n}{L_{\max}}\right)^{1/(d+1)}\right)$ for spreading to all nodes w.h.p., and consequently also in expectation. Barring a logarithmic factor, this shows that such a random policy is as good as any other (possibly omniscient) policy for grids. In order to derive this lower bound, we first need the following lemma from the theory of first-passage percolation \cite{kesten93:speed}, which lets us control the extent to which infection on an infinite grid has spread at time $t$:

\begin{lemma}
\label{lem:gridpercolation}
Let $(\tilde{Z}(t))_{t \geq 0} \in \{0,1\}^{\mathbb{Z}^d}$ represent a static/basic infection spread process on the infinite $d$-dimensional lattice $\mathbb{Z}^d$ starting at node $(0,0,\ldots,0)$ at time $0$. Then, there exist positive constants $l, c_3, c_4$ such that for $t\geq 1$:
$$\mathbb{P}[\mathcal{N}(\tilde{Z}(t)) > t^d l^d] \leq c_1 t^{2d}e^{-c_2\sqrt{t}}. $$
\end{lemma}

\begin{proof}
Let
\[\tilde{B}(t) \bydef \{v \in \mathbb{Z}^d: \tilde{Z}_v(t) = 1\}  \subset \mathbb{Z}^d \; (\subset \mathbb{R}^d)\]
be the set of infected nodes at time $t$ in $\tilde{Z}$. We use the following version of a result, from percolation on lattices with exponentially distributed edge passage times, about the `typical shape' of $\tilde{B}(t)$ \cite{kesten93:speed}:

\noindent{\em (Theorem 2 in \cite{kesten93:speed})} { \em There exists a fixed (\emph{i.e.} not depending on $t$) cube $B_0 =
  \left[-\frac{l}{2},\frac{l}{2}\right]^d \subset \mathbb{R}^d$, and constants $c_1, c_2 > 0$, such that for $t \geq 1$,}
\begin{align} 
  \label{eqn:shapethm}
  \mathbb{P}\left[\tilde{B}(t) \subset tB_0\right] \geq 1 - c_1 t^{2d}e^{-c_2\sqrt{t}}.
\end{align}
It follows from (\ref{eqn:shapethm}) that for $t \geq 1$,
\begin{align*}
\mathbb{P}[\mathcal{N}(\tilde{Z}(t)) > t^d l^d] &= \mathbb{P}[|\tilde{B}(t)| > t^d l^d] \\
&\leq \mathbb{P}[\tilde{B}(t)\nsubseteq tB_0]
\leq c_1 t^{2d}e^{-c_2\sqrt{t}}.
\end{align*}
\end{proof}

Lemma \ref{lem:gridpercolation} allows us to control the growth of individual infected clusters; this is analogous to the dominating spread process (growing at rate $2\beta$) for line graphs. Using this, we now obtain a lower bound on the spreading-time.
\begin{proof}[Proof of Theorem \ref{thm:gridconv}]
We introduce a (dominating) counting process $(\tilde{S}(t))_{t \geq 0}$ (Fig. \ref{fig:griddom}), as follows:
\begin{itemize}
\item $\forall\,t\geq 0$, $\tilde{S}(t)$ consists of an integer number ($\tilde{C}_t$) of \emph{clusters}, where $(\tilde{C}_t)_{t \geq 0}$ is a Poisson process with intensity $L_{\max}(n)$, and $\tilde{C}_0 = 1$ (i.e., an `initial' infected node).
\item Each cluster grows as an independent copy of the intrinsic spreading process on an exclusive infinite $d$-dimensional grid $\mathbb{Z}^d$ starting at $(0,0,\ldots,0)$.
\end{itemize}
Note that in the process $\tilde{S}$, the growth of each cluster follows the intrinsic spreading dynamics in a $d$-dimensional grid graph. A standard coupling argument shows that $\forall t \geq 0$, the total number of points in $\tilde{S}(t)$ (denoted by $\tilde{N}_t$) stochastically dominates that in $S(t)$ -- this is essentially due to (a) cluster 'seeding' at the highest possible rate $L_{\max}(n)$, and (b) the absence of multiple infections incident at any single node (Fig. \ref{fig:griddom}). Let $\tilde{T} \bydef \inf \{t \geq 0: \tilde{N}_t= n \}$ be the time when the number of points in $\tilde{S}(\cdot)$ first hits $n$. Then we have:
\begin{align}
  \mathcal{N}(S(t)) \leq_{st} \tilde{N}_t \;
  \Rightarrow \; \tilde{T} &\leq_{st} T. \label{eqn:stdomgrid}
\end{align}  
\noindent Let us denote by $\tilde{X}_i(s)$ the size of the $i$th created cluster of $\tilde{S}(\cdot)$ at time $s \geq T_i$. Then, for $t \geq 0$, we have
\begin{align*}
\left(\bigcap_{i=0}^{2et} \{\tilde{X}_i(t+T_i) < t^d l^d\}\right)
  &\bigcap\Bigg(\{\tilde{C}_t < 2eL_{\max}(n)t \}\Bigg)\\
  & \subseteq \{\tilde{N}_t < 2eL_{\max}(n)l^d t^{d+1}\},
\end{align*}
Now each random variable $\tilde{X}_i(t+T_i)$ is distributed as the number of infected nodes in a static infection process on an infinite grid at time $t$. Thus, using Lemma \ref{lem:gridpercolation} and a standard Chernoff bound for $\tilde{C}_t \sim$ Poisson($tL_{\max}(n)$), we can write: 
\begin{align*}
  &\mathbb{P}\left[\tilde{N}_t \geq (2eL_{\max}(n)l^d) t^{d+1}\right] \\
  &\leq\mathbb{P}\left[\tilde{C}_t \geq 2eL_{\max}(n)t\right] +
  \sum_{i=1}^{2eL_{\max}(n)t}\mathbb{P}\left[\tilde{X}_i(t+T_i) \geq t^d l^d\right] \\
  &\leq (2e)^{-L_{\max}(n)t} + 2eL_{\max}(n)t \cdot c_3 t^{2d} e^{-c_4\sqrt{t}}\\
  &= O(L_{\max}(n)e^{-c_4\sqrt{t}}). 
\end{align*}
With the stochastic dominance (\ref{eqn:stdomgrid}), this forces:
\begin{align}
\mathbb{P}\Bigg[T\leq&\left(\frac{n}{2eL_{\max}(n)l^d}\right)^{1/(d+1)}\Bigg] \nonumber \\
&\leq\mathbb{P}\left[\tilde{T}\leq\left(\frac{n}{2eL_{\max}(n)l^d}\right)^{1/(d+1)}\right] \nonumber \\
&=\mathbb{P}\left[\tilde{N}_{\left(\frac{n}{2eL_{\max}(n)l^d}\right)^{1/(d+1)}} \geq
    n \right] \nonumber \\
    &= O\left(e^{-c_2 \left(\frac{n}{L_{\max}(n)}\right)^{1/(2d+2)}}\right), \label{eqn:probdecay}
\end{align}
for the appropriate $c_2$, establishing the first part of the theorem. To see how this implies the second part, note that the estimate (\ref{eqn:probdecay}), together with the fact
that $L_{\max}(n) = O(n^{1-\epsilon})$ and the Borel-Cantelli lemma, gives us
\begin{align*}
&\mathbb{P}\left[\tilde{T} \leq \left(\frac{n}{2eL_{\max}(n)l^d}\right)^{1/(d+1)} \mbox{ for finitely many $n$}\right] = 1, \\
&\Rightarrow \liminf_{n \to \infty} \frac{\tilde{T}}{\left(n/L_{\max}(n)\right)^{1/(d+1)}}\stackrel{a.s.}{\geq} c_4 \bydef \frac{1}{(2el^d)^{1/(d+1)}} > 0
\end{align*}
By Fatou's lemma, we have:
\begin{align*}
\liminf_{n \to \infty}& \mathbb{E}\left[\frac{\tilde{T}}{\left(n/L_{\max}(n)\right)^{1/(d+1)}}\right]\\
&\geq \mathbb{E}\left[\liminf_{n \to \infty} \frac{\tilde{T}}{\left(n/L_{\max}(n)\right)^{1/(d+1)}}\right]\geq c_4 > 0.
\end{align*}
Thus proving $\mathbb{E}[T] \geq \mathbb{E}[{\tilde{T}}] =
\Omega\left(\left(n/L_{\max}(n)\right)^{1/(d+1)}\right)$.
\begin{figure}
\centering \includegraphics[width=0.75\columnwidth]{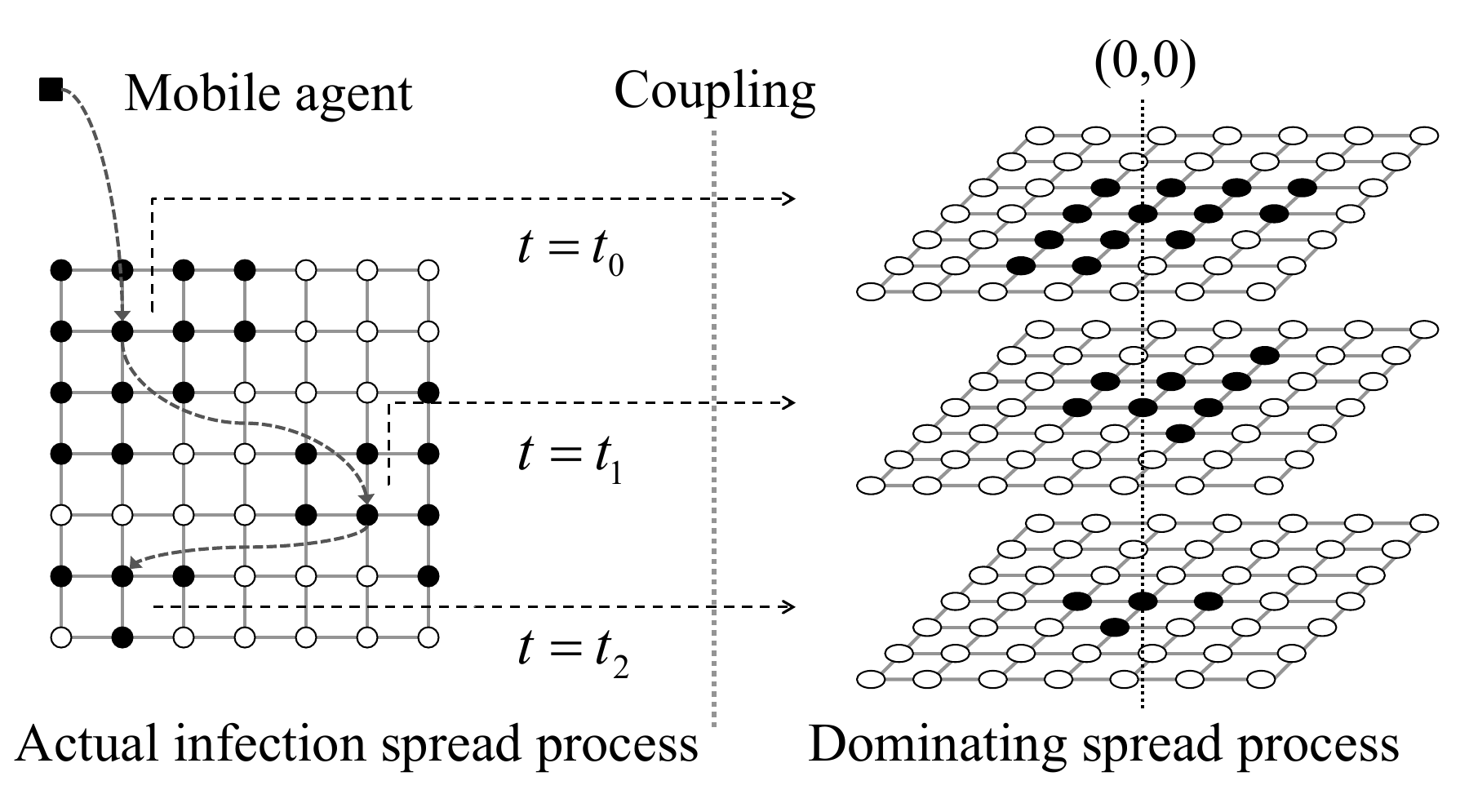}
  \caption{The grid graph: Coupling infection spreading with mobility to a dominating `cluster-growth' process}
  \label{fig:griddom}
\end{figure}
\end{proof}

\subsection{Geometric Random Graphs}
We finally prove the upper and lower bounds for the \emph{Geometric Random Graph (RGG)}. Recall that an RGG $G_n(r_n)$ is a family of random graphs wherein $n$ nodes are picked i.i.d. uniformly in $[0,1] \times [0,1]$. Two nodes $x, y$ have an edge iff $||x-y|| \leq r_n$, where $r_n$ is called the \emph{coverage radius}. 

It is known that when the coverage radius $r_n$ is above a critical threshold of $\sqrt{\log n/\pi}$, the RGG is connected with high probability \cite{gupta:connectivity}. We now obtain two results that show that random spreading on RGGs in this critical connectivity regime is optimal upto logarithmic factors. First, we show with high probability that random spreading finishes in time $O(\sqrt[3]{n} \log n)$, and follow it up by showing that with high probability, no other policy can better this order (up to a logarithmic factor). This parallels the earlier results regarding the random external spreading policy on grids.

\begin{proof}[Proof of Theorem \ref{thm:rggach}]
Divide the unit square $[0,1] \times [0,1]$ into square \emph{tiles} of side length $r_n/\sqrt{5}$ each; there are thus
$5/r_n^2$ such tiles, say $k_1, \ldots, k_{5/r_n^2}$. If $n$ points are thrown uniformly randomly into $[0,1] \times [0,1]$, then, with $\mathcal{E}$ denoting the event that some tile is empty:
\begin{align}
\mathbb{P}\left[\mathcal{E}\right]&\leq\frac{5}{r_n^2} \mathbb{P}\left[ \mbox{tile 1 empty}\right]=\frac{5}{r_n^2} \left(1-\frac{r_n^2}{5}\right)^n\nonumber\\
&\leq \frac{n}{\log n} \exp(-\log n)= \frac{1}{\log n} \stackrel{n \to \infty}{\longrightarrow} 0. \label{eqn:whp} 
\end{align}
By construction, note that the maximum distance between points in two (horizontally or vertically) adjacent tiles is exactly $r_n$. Hence, two nodes in adjacent tiles are always connected by an edge. Also, a node in a tile is not connected to any node in a tile at least three hops away.

Let $\tilde{n}\triangleq nL_{\min}^2$. If we now divide $[0,1] \times [0,1]$ into (bigger) square \emph{chunks} of side length $1/\sqrt[6]{\tilde{n}}$ each, there are $\sqrt[3]{\tilde{n}}$ such square chunks, each containing a $\frac{\sqrt{5}}{r_n \tilde{n}} \times \frac{\sqrt{5}}{r_n \tilde{n}}$ grid of square tiles. In the case where no tile is empty, it follows from the arguments in the preceding paragraph that the diameter $D$ of the subgraph induced within each chunk is:
\begin{align*}
D\leq \frac{2\sqrt{5}}{r_n\sqrt[6]{\tilde{n}}} \leq 
\frac{2}{\sqrt{(1+\gamma)\log n}}\sqrt[3]{\frac{n}{L_{\min}}}.
\end{align*}
Choosing $n\geq e^{4/1+\gamma}$, an application of Theorem \ref{thm:genach} in this case shows that for random external-spreading: 
$$\mathbb{E}\left[T \vert \mathcal{E}\right]\leq \sqrt[3]{n/L_{\min}}\cdot(\log n +1).$$
Note also that for any graph, using the random external-spreading strategy, we have $\mathbb{E}[T]\leq\frac{n}{L_{\min}}$. Thus if $\gamma\geq 2/3$, then combining with (\ref{eqn:whp}), we have:
$$\mathbb{E}[T]\leq 2\sqrt[3]{n/L_{\min}}\log n .$$

Further, let $\delta\triangleq\frac{2}{3}\log_n\left[\frac{n}{L_{\min}}\right]$ -- then given $\mathcal{E}$, we have that each subset in the partition has size $\geq n^{\delta}$. Now, for any $\gamma > 0$, and choosing $\kappa>\left(1+\frac{3\gamma}{1+\delta}\right)$, we have from the concentration in Theorem $\ref{thm:genach}$ that:
\begin{equation*}
\mathbb{P}\left[T\geq \kappa \sqrt[3]{n/L_{\min}} \log n|\mathcal{E}\right]\leq  \frac{1}{n^{\gamma}}.
\end{equation*}
Combining this with equation (\ref{eqn:whp}), we get: 
\begin{equation*}
\mathbb{P}[T\geq \kappa \sqrt[3]{n/L_{\min}} \log n] \leq \frac{2}{n^{\gamma}}.
\end{equation*}
This completes the proof.
\end{proof}

Consider an infinite planar grid with additional one-hop diagonal
edges, i.e. $G = (V, E)$ where $V = \mathbb{Z}^2$, $E = \{(x,y) \in\mathbb{Z}^2: ||x-y||_\infty \leq 1 \}$. Let an infection process $(S(t))_{t \geq 0}$ start from $0 \in \mathbb{Z}^2$ at time $0$ according to the standard static spread dynamics, i.e. with each edge propagating infection at an exponential rate $\beta$, and let $I(t)$ denote the set of infected nodes at time $t$. The following key lemma helps control the size of $I(t)$, \emph{i.e.} the extent of infection at time $t$:
\begin{lemma}
\label{lem:shape}
There exists $c_1 > 0$ such that for any $c_2 > 0$ and $t$ large enough:
\begin{align*}
 \mathbb{P}\left[\exists x \in I(t): ||x||_\infty \geq (c_1 \beta + c_2)t \right] =O\left((c_1 \beta + c_2)t \cdot e^{-c_2 t}\right). 
\end{align*}
\end{lemma}

\begin{proof}
\begin{align*}
\mathbb{P}&[\exists x \in I(t): ||x||_\infty\geq ct] \\ 
&\leq \mathbb{P}[\exists v \in \mathbb{Z}^2:||v||_\infty = \lfloor ct \rfloor, T(v) \leq t ] \\
&\leq \sum_{v \in \mathbb{Z}^2: ||v||_\infty = \lfloor ct \rfloor} \mathbb{P}[\exists \mbox{ a path } r:0 \rightarrow v, T(r) \leq t].
\end{align*}
Observe that for any $v$ with $||v||_\infty = \lfloor ct \rfloor$ and any path of edges $r$ from 0 to $v$, there must exist $\lfloor ct \rfloor + 1$ nodes $v_0 = 0, v_1, \ldots, v_{\lfloor ct \rfloor}$ on the path $r$ such that $||v_i||_\infty \leq \lfloor ct \rfloor$ and $||v_{i+1} - v_i||_\infty = 1$. Indeed, each edge on a path can increase the $||\cdot||_\infty$ distance from $0$ by at most $1$. Therefore, continuing the above chain of inequalities, we have:
\begin{align}
&\sum_{\{v: ||v||_\infty = \lfloor ct \rfloor\}} \sum_{v_0,\ldots,v_{\lfloor ct \rfloor}} \mathbb{P}\left[\exists\mbox{ a path } r:0 \rightarrow v \mbox{ passing } \right. \nonumber\\
&\left. \mbox{ successively through the } v_i, T(r) \leq t \right] \nonumber \\    
&\leq \sum_{\{v: ||v||_\infty = \lfloor ct \rfloor\}}
\sum_{v_0,\ldots,v_{\lfloor ct \rfloor}} \mathbb{P}\Bigg[\exists\mbox{ a path } r \mbox{ passing } \nonumber \\
&\mbox{successively through the }v_i, \sum_{i=0}^{\lfloor ct \rfloor -1} T(v_i,v_{i+1}) \leq t \Bigg] \nonumber\\
&\leq \sum_{\{v: ||v||_\infty = \lfloor ct \rfloor\}}
\sum_{v_0,\ldots,v_{\lfloor ct \rfloor}} \mathbb{P}\left[\sum_{i=0}^{\lfloor ct \rfloor -1} T(v_i,v_{i+1}) \leq t \right], \label{eqn:probbound}
\end{align} 
where the second sum runs throughout over all $v_i$ with $v_0 = 0$, $||v_i||_\infty \leq \lfloor ct \rfloor$ and $||v_{i+1}-
v_i||_\infty = 1$, and $T(x,y)$ represents the infection passage time from node $x$ to node $y$. Letting $T'(v_i,v_{i+1})$ be random variables identically distributed as $T(v_i,v_{i+1})$ but \emph{independent} for $i = 1, \ldots, \lfloor ct \rfloor - 1$, we can write, for $\psi > 0$,
\begin{align*}
\sum_{v_0,\ldots,v_{\lfloor ct \rfloor}} &\mathbb{P}\left[\sum_{i=0}^{\lfloor ct \rfloor -1} T(v_i,v_{i+1})\leq t \right]\\  
&=\sum_{v_0,\ldots,v_{\lfloor ct \rfloor}} \mathbb{P}\left[\sum_{i=0}^{\lfloor ct \rfloor -1} T'(v_i,v_{i+1})\leq t \right] \\
&\leq e^{\psi t}\sum_{v_0,\ldots,v_{\lfloor ct \rfloor}} \prod_{i=0}^{\lfloor ct \rfloor -1} \mathbb{E}\left[e^{-\psi T'(v_i,v_{i+1})}\right] \\
&= e^{\psi t} \left(\sum_{\{u: ||u||_\infty = 1\}}
\mathbb{E}\left[e^{-\psi T'(0,u)}\right]\right)^{\lfloor ct \rfloor}.
\end{align*}
In the last step of the above display, we have successively summed over $v_{\lfloor ct \rfloor}, v_{\lfloor ct \rfloor -1}, \ldots,v_0$, and have used the fact that infection spread times are translation-invariant, \emph{i.e.} for any $x,y,a\in\mathbb{Z}^2$,
\begin{align*}
T'(x,y) &\stackrel{d}{=} T(x,y)\stackrel{d}{=} T(x+a,y+a)\stackrel{d}{=} T'(x+a,y+a). 
\end{align*} 

For any $u \in \mathbb{Z}^2$ such that $u$ is a neighbor of $0$ (\emph{i.e.} $||u||_\infty = 1$), we must have $T(0,u) \geq \min_{w:||w||_\infty = 1} t((0,w))$, where $t(e) \sim$ Exp($\mu$) is the travel time of the infection across edge $e \in E$. Since the number of neighbors of $0$ in $G$ is exactly $8$ ($4$ up-down/left-right and $4$ diagonal), $T(0,u)$ stochastically dominates an exponential random variable with parameter $8\mu$. Thus, defining $ \hat{T} \sim \mbox{Exp}(8\mu)$, we have:
\begin{align}
\mathbb{E}&\left[e^{-\psi T'(u,v)}\right] \leq \mathbb{E}\left[e^{-\psi \hat{T}}\right] = \left( 1 + \frac{\psi}{8\mu} \right)^{-1},\\
\Rightarrow \; &e^{\psi t} \left(\sum_{\{u: ||u||_\infty = 1\}}
    \mathbb{E}\left[e^{-\psi T'(0,u)}\right]\right)^{\lfloor ct
      \rfloor} \nonumber \\   
    &\leq e^{\psi t} \left( 8 \left( 1 +
    \frac{\psi}{8\mu} \right)^{-1}\right)^{\lfloor ct \rfloor}. \label{eqn:mgfbound}
\end{align}
Setting $\psi = 8\mu (8e-1)$ so that $8(1 + \psi/\mu)^{-1} =
  e^{-1}$, equation (\ref{eqn:mgfbound}) becomes:
\begin{align*}
    e^{\psi t} \Bigg(\sum_{\{u: ||u||_\infty = 1\}}
	\mathbb{E}&\left[e^{-\psi T'(0,u)}\right]\Bigg)^{\lfloor ct
	  \rfloor}\\
	& \leq e^{8\mu (8e-1) t} \cdot e^{-ct + 1}. 
\end{align*} 
Finally, letting $c_1 = 8(8e-1)$ and $c = c_1\mu + c_2$, we obtain the desired result from (\ref{eqn:probbound}) and the above:
  \begin{align*}
    \sum_{\{v: ||v||_\infty = \lfloor ct \rfloor\}}
    &\sum_{v_0,\ldots,v_{\lfloor ct \rfloor}}
    \mathbb{P}\left[\sum_{i=0}^{\lfloor ct \rfloor -1} T(v_i,v_{i+1})
      \leq t \right]\\ 
    &\leq  |\{v: ||v||_\infty = \lfloor ct \rfloor\}| \cdot e^{-c_2 t + 1}\\
    &\leq (4ct)\cdot e^{-c_2 t + 1} \\
    &= O\left((c_1 \mu + c_2)t \cdot e^{-c_2 t}\right).
  \end{align*}  
\end{proof}

Using Lemma \ref{lem:shape}, we can derive a converse result for the geometric random graph, which we present in Theorem \ref{thm:rggconv}. As the proof techniques are similar to those presented before, we present only a sketch of the proof for this result.

\begin{proof}[Proof sketch of Theorem \ref{thm:rggconv}]
The method of approach is along the lines of that used to prove Theorem \ref{thm:gridconv} along with certain geometric considerations for the case of the random geometric graph. We introduce a spreading process that spreads `faster' than $\pi$, and show using Lemma \ref{lem:shape} that even this process must take at least the claimed amount of time to spread. For ease of exposition, we break the proof into two steps:

\emph{Step 1:} Divide the unit square $[0,1] \times [0,1]$ row and column-wise into $r_n \times r_n$ \emph{tiles}; there are thus $1/r_n^2$ tiles, say $k_1,\ldots,k_{1/r_n^2}$. By standard balls-and-bins arguments, with the $n$ nodes thrown randomly into $n/\log n$ tiles, each tile receives a maximum of $O(\log n)$ nodes with probability $1-o(1)$.

\emph{Step 2:} Within the event in step 1, we introduce the following associated spreading process which, via coupling arguments, can be shown to dominate the spread due to $\pi$ at each time $t$: first, take each tile to be the vertex of a square grid where adjacent diagonals are connected. Also, set the rate of infection spread on every edge to be Exp($\mu \log^2 n$). This effectively upper-bounds the best rate of spread among neighboring tiles. Create a dominating process by creating non-interfering clusters at a Poisson rate 1, with each cluster growing independently on an infinite square grid with diagonal edges and the above spread rate. Lemma \ref{lem:shape} shows that w.h.p., by time $t$, $O(t)$ clusters are formed, and each cluster has at most $O(t^2\log^4 n)$ nodes. Thus it takes at least $O\left(\frac{\sqrt[3]{n}}{\log^{4/3}n}\right)$ time for spreading to spreading w.h.p.
\end{proof}

\section{Conclusion}
We have modeled and analyzed the spread of epidemic processes on graphs when assisted by external agents. For general graphs, we have provided upper bounds on the spreading time due to external-infection with bounded virulence for random and greedy infection policies; these bounds are in terms of the diameter and the conductance of the graph. On the other hand, for certain spatially-constrained graphs such as grids and the geometric random graph, we have derived corresponding lower bounds: these indicate that random external-infection spreading is order-optimal up to logarithmic factors (and greedy is order-optimal) in such scenarios. Finally, we have discussed applications of our result to graphs with long-range edges and/or mobile agents.

\section*{Acknowledgements}
This work was partially supported by NSF Grants CNS-1320175, CNS-1017525, CNS-0721380 and Army Research Office Grant W911NF-11-1-0265. We also thank the anonymous reviewers for their suggestions, which helped greatly in improving the presentation of this paper.

\end{document}